\documentclass[11pt]{article}

\usepackage{lipsum}
\usepackage{amsfonts}
\usepackage{graphicx}
\usepackage{epstopdf}
\usepackage{algorithm}
\usepackage{algorithmic}
\ifpdf
  \DeclareGraphicsExtensions{.eps,.pdf,.png,.jpg}
\else
  \DeclareGraphicsExtensions{.eps}
\fi

\usepackage[left=2cm,right=2cm,top=2cm,bottom=2cm]{geometry}

\usepackage{tabu}       

\usepackage{nicefrac}       
\usepackage{microtype}      
\usepackage{lipsum}
\usepackage[utf8]{inputenc} 
\usepackage[T1]{fontenc}    
\usepackage{braket}
\usepackage{graphicx}  
\usepackage{amsmath}  

\usepackage{hyperref}
\usepackage{xcolor}
\hypersetup{
    colorlinks,
    linkcolor={red!50!black},
    citecolor={blue!50!black},
    urlcolor={blue!80!black}
}

\usepackage{multirow}
\usepackage{url}
\usepackage[all]{xy}
\usepackage{microtype}
\usepackage{graphicx}
\usepackage{subcaption}
\usepackage{ stmaryrd }
\usepackage{ upgreek }
\usepackage{booktabs} 
\usepackage{ dsfont }
\usepackage{amsfonts}       
\usepackage{amsthm}
\usepackage{enumerate}
\usepackage{forloop}
\usepackage{comment}

\usepackage{float}
\usepackage{tabularx}

\usepackage{ latexsym }
\usepackage{pst-node}
\usepackage{tikz-cd}
\usepackage{mathabx}
\usepackage{esvect}
\usepackage{tikz}
\usepackage{paralist}
\usetikzlibrary{arrows}

\usepackage{authblk}

\usepackage{mathtools}

\theoremstyle{plain}
\newtheorem{proposition}{Proposition}
\newtheorem{theorem}{Theorem}
\newtheorem{lemma}{Lemma}

\theoremstyle{definition}
\newtheorem{definition}{Definition}
\newtheorem{remark}{Remark}

\usepackage[capitalize,noabbrev]{cleveref}

\graphicspath{{./fig/}}

\title{A Gradient Sampling Algorithm for Stratified Maps with Applications to Topological Data Analysis}

\author[1]{Jacob Leygonie\thanks{\texttt{jacob.leygonie@maths.ox.ac.uk}}}
\author[2]{Mathieu Carrière\thanks{\texttt{mathieu.carriere@inria.fr}}}
\author[3]{Théo Lacombe\thanks{\texttt{theo.lacombe@univ-eiffel.fr}}}
\author[4]{Steve Oudot\thanks{\texttt{steve.oudot@inria.fr}}}
\affil[1]{Mathematical Institute, Oxford University, UK}
\affil[2]{DataShape, Université Côte d'Azur, Inria, France}
\affil[3]{LIGM, Université Gustave Eiffel, France.}
\affil[4]{DataShape, Université Paris-Saclay, CNRS, Inria, Laboratoire de Mathématiques d’Orsay, France.}
\date{}

\makeatletter
\newcommand*{\addFileDependency}[1]{
  \typeout{(#1)}
  \@addtofilelist{#1}
  \IfFileExists{#1}{}{\typeout{No file #1.}}
}
\makeatother

\newcommand\blfootnote[1]{%
  \begingroup
  \renewcommand\thefootnote{}\footnote{#1}%
  \addtocounter{footnote}{-1}%
  \endgroup
}

\newcommand{\R}{\mathbb{R}}
\newcommand{\N}{\mathbb{N}}
\newcommand{\Stratum}{X}
\newcommand{\Stratif}{\mathcal{X}}
\newcommand{\lipconst}{L}
\newcommand{\co}{\overbar{\mathrm{co}}}
\newcommand{\Goldstein}{\partial_\epsilon f}
\newcommand{\AppGoldstein}{\tilde{\partial}_\epsilon f}
\newcommand{\Perm}{\Sigma_n}
\newcommand{\StratumPi}{\mathcal{S}_{\pi}}

\newcommand{\Transp}{\mathcal{T}_n}
\newcommand{\inv}{\mathrm{inv}}
\newcommand{\adist}{\hat{d}}
\newcommand{\AppAppGoldstein}{\hat{\partial}_\epsilon f}
\newcommand{\persmap}{\mathrm{PH}}
\newcommand{\Pers}{\mathrm{Pers}}
\newcommand{\SComplex}{K}
\newcommand{\simplex}{\sigma}
\newcommand{\V}{V}
\newcommand{\Loss}{f}
\newcommand{\Barc}{\mathbf{Bar}}

\newcommand{\overbar}{\overline}

\newcommand{\repolink}{https://github.com/tlacombe/topt}

\newcommand{\defeq}{\vcentcolon=}

\newcommand{\AlgoGradientDescent}{\mathbf{SGS}}
\newcommand{\AlgoUpdateStep}{\mathbf{UpdateStep}}
\newcommand{\AlgoUpdateStepSimple}{\mathbf{SimpleUpdateStep}}
\newcommand{\AlgoGeneralizedGradient}{\mathbf{ApproxGradient}}
\newcommand{\AlgoMakeDiff}{\mathbf{MakeDifferentiable}}

\newcommand{\ControlConst}{C}
\newcommand{\DiffSet}{\mathcal{D}}
\newcommand{\DiffOracle}{\mathbf{DiffOracle}}
\newcommand{\SampleOracle}{\mathbf{SampleOracle}} 
\newcommand{\ApproxSampleOracle}{\mathbf{ApproxSampleOracle}} 
\newcommand{\proxypoint}{\tilde{\elt}_{\Stratum}}
\newcommand{\consta}{a}
\newcommand{\CardStrata}{N}

\newcommand{\mapper}{{\rm Map}}

\newcommand{\elt}{x}
\newcommand{\filt}{x}

\DeclareMathOperator*{\argmin}{\mathrm{argmin}}

\begin{document}

\maketitle

\begin{abstract}
  We introduce a novel gradient descent algorithm extending the well-known Gradient Sampling methodology to the class of stratifiably smooth objective functions, which are defined as locally Lipschitz functions that are smooth on some regular pieces---called the strata---of the ambient Euclidean space. For this class of functions, our algorithm achieves a sub-linear convergence rate. We then apply our method to objective functions based on the (extended) persistent homology map computed over lower-star filters, which is a central tool of Topological Data Analysis. For this, we propose an efficient exploration of the corresponding stratification by using the Cayley graph of the permutation group. Finally, we provide benchmark and novel topological optimization problems, in order to demonstrate the utility and applicability of our framework.
\end{abstract}

\blfootnote{This research was conducted while Théo Lacombe was affiliated to DataShape, Université Paris-Saclay, CNRS, Inria, Laboratoire de Mathématiques d’Orsay, France.}

\tableofcontents{}


\section{Introduction}
\label{sec:intro}

\subsection{Motivation and related work}
\label{sec:related_work}
In its most general instance nonsmooth, non convex, optimization seek to minimize a locally Lipschitz objective or loss function $f:\R^n \rightarrow \R$. Without further regularity assumptions on~$f$, most algorithms---such as the usual Gradient Descent with learning rate decay, or the Gradient Sampling method---are only guaranteed to produce iterates whose subsequences are asymptotically stationary, without explicit convergence rates. 
Meanwhile, when restricted to the class of min-max functions (like the maximum of finitely many smooth maps), stronger guarantees such as convergence rates can be obtained~\cite{helou2017local}. 
This illustrates the common paradigm in nonsmooth optimization: the richer the structure in the irregularities of $f$, the better the guarantees we can expect from an optimization algorithm. Note that there are algorithms specifically tailored to deal with min-max functions, e.g.~\cite{bertsekas1975nondifferentiable}.

Another example are bundle methods~\cite{fuduli2004dc,fuduli2004minimizing,haarala2007globally}. They consist, roughly, in constructing successive linear approximations of~$f$ as a proxy for minimization. Their convergence guarantees are strong, especially when an additional {\em semi-smoothness} property of~$f$~\cite{bihain1984optimization,mifflin1977algorithm} can be made. Other types of methods, like the variable metric ones, can also benefit from the semi-smoothness hypothesis~\cite{vlvcek2001globally}. In many cases, 
convergence properties of the algorithm are not only dependent on the structure on $f$, 
but also on the amount of information about $f$ that can be computed in practice. 
For instance, the bundle method~\cite{lukvsan1998bundle} assumes that the Hessian matrix, when defined locally, can be computed. For accounts of the theory and practice in nonsmooth optimization, we refer the interested reader to~\cite{bagirov2014introduction,kiwiel2006methods,shor2012minimization}.

The ability to cut~$\R^n$ in well-behaved pieces where~$f$ is regular, is another type of important structure. Examples, in increasing order of generality, are semi-algebraic, (sub)analytic, definable, tame (w.r.t. an o-minimal structure), and Whitney stratifiable functions~\cite{bolte2007clarke}. For such objective functions, the usual gradient descent (GD) algorithm, or a stochastic version of it, converges to stationary points~\cite{davis2020stochastic}. In order to obtain further theoretical guarantees such as convergence rates, it is necessary to design optimization algorithms specifically tailored for regular maps, since they enjoy stronger properties, e.g., tame maps are semi-smooth~\cite{ioffe2009invitation}, and the generalized gradients of Whitney stratifiable maps are closely related to the (restricted) gradients of the map 
along the strata~\cite{bolte2007clarke}. Besides, strong convergence guarantees can be obtained under the Kurdyka--{\L}ojasiewicz assumption~\cite{attouch2013convergence,noll2014convergence}, which includes the class of semi-algebraic maps. Our method is related to this line of work, in that we exploit the strata of~$\R^n$ in which~$f$ is smooth. 

The motivation of this work stems from Topological Data Analysis (TDA), where geometric objects such as graphs are described by means of computable and topological descriptors. Persistent Homology (PH) is one such descriptor, and has been successfully applied in various areas such as neuroscience~\cite{dabaghian2014reconceiving,bendich2016persistent}, material sciences~\cite{hiraoka2016hierarchical,townsend2020representation}, signal analysis~\cite{perea2015sliding,umeda2017time}, shape recognition~\cite{li2014persistence}, or machine learning \cite{chen2019topological,carriere2020perslay}. 

Persistent Homology describes graphs, and more generally simplicial complexes, over~$n$ nodes by means of a signature called the \emph{barcode}, or \emph{persistence diagram}~$\persmap(\filt)$. Here~$\filt$ is a \emph{filter function}, that is a function on the nodes, which we view as a vector in~$\R^n$. 
Loosely speaking,~$\persmap(\filt)$ is a finite multi-set of points in the upper half-plane $\{(b,d)\in \R^2,\ d\geq b \}$ that encodes topological and geometric information about the underlying simplicial complex and the function~$\filt$. 

Barcodes form a metric space $\Barc$ when equipped with the standard metrics of TDA, the so-called {\em bottleneck} and {\em Wasserstein} distances, and the persistence map~$\persmap: \R^n \rightarrow \Barc$ is locally Lipschitz~\cite{cohen2007stability,cohen2010lipschitz}. However~$\Barc$ is not Euclidean nor a smooth manifold, thus hindering the use of these topological descriptors in standard statistical or machine learning pipelines. Still, there exist natural notions of differentiability for maps in and out of~$\Barc$~\cite{leygonie2019framework}. 
In particular, the persistence map~$\persmap: \R^n \rightarrow \Barc$ restricts to a locally Lipschitz, smooth map on a stratification of~$\R^n$ by polyhedra. 
If we compose the persistence map with a smooth and Lipschitz map~$\V: \Barc\rightarrow \R$, the resulting objective (or loss) function
$$
\xymatrix{
\Loss: \R^n \ar[rr]^{\persmap}&& \Barc  \ar[rr]^{\V}&&\R }
$$
is itself Lipschitz and smooth on the various strata. From~\cite{davis2020stochastic}, and as recalled in~\cite{carriere2020note}, classical (Stochastic) Gradient Descent on~$\Loss$ asymptotically converges to stationary points. Similarly, the Gradient Sampling (GS) method asymptotically converges. See~\cite{solomon2020fast} for an application of GS to topological optimization. 

Nonetheless, it is important to design algorithms that take advantage of the structure in the irregularities of the persistence map~$\persmap$, in order to get better theoretical guarantees. For instance, one can locally integrate the gradients of~$\persmap$---whenever defined---to stabilize the iterates~\cite{solomon2020fast}, or add a regularization term to~$\Loss$ that acts as a smoothing operator~\cite{corcoran2020regularization}. In this work, we rather exploit the stratification of~$\R^n$ induced by~$\persmap$, as it turns out to be easy to manipulate. We will show in particular that we can efficiently access points~$\filt'$ located in neighboring strata of the current iterate $\filt$, as well as estimate the distance to these strata. 

For this reason, we believe that persistent homology-based objective functions~$\Loss$ form a rich playground for nonsmooth optimization, 
with many applications in point cloud inference~\cite{gameiro2016continuation}, surface reconstruction~\cite{bruel2020topology}, shape matching~\cite{poulenard2018topological}, graph classification~\cite{hofer2020graph,yim2021optimisation}, topological regularization for generative models~\cite{moor2019topological, hofer2019connectivity,gabrielsson2020topology}, image segmentation~\cite{hu2019topology,clough2019explicit}, or dimensionality reduction~\cite{kachan2020persistent}, to name a few. 

\begin{figure}
    \includegraphics[width=0.58\textwidth]{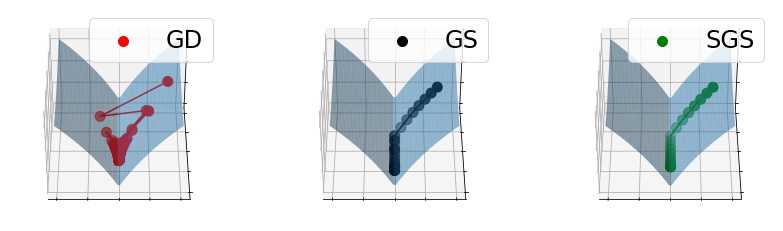}
    \includegraphics[width=0.35\textwidth]{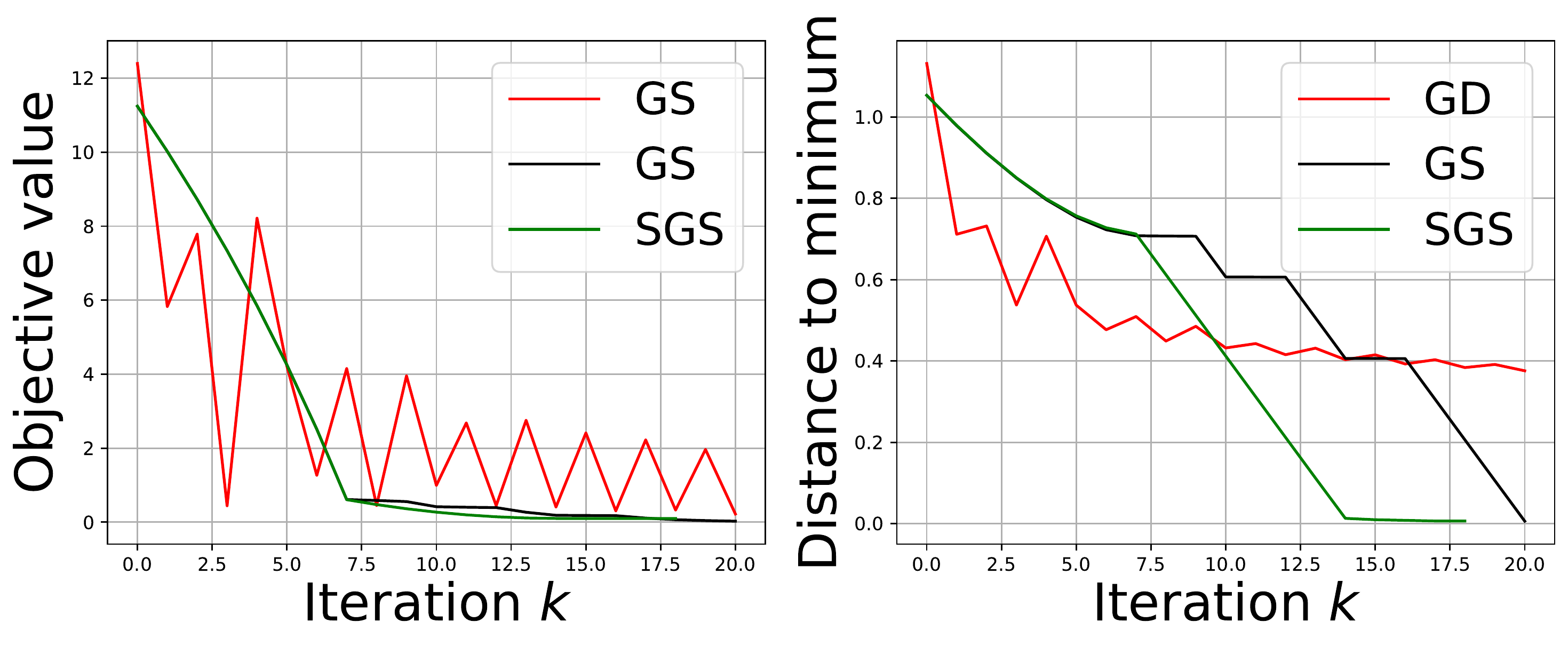}
    \caption{A proof-of-concept comparison between different gradient descent techniques. The objective function $\Loss : (z_1,z_2) \in \R^2 \to 10 \log(1 + |z_1|) + z_2^2 \in  \R$ (blue surface) attains its minimum at $\elt_* = (0,0)$ and is not smooth along the line~$\{z_1=0\}$. In particular, $\| \nabla \Loss \| > 1$ around~$\elt_*$, thus the gradient norm cannot be used as a stopping criterion. The traditional GD, for which updates are given by~$\elt_{k+1} = \elt_k - \frac{\lambda_0}{k+1} \nabla \Loss(\elt_k)$, oscillates around~$\{z_1=0\}$ due to the non-smoothness of~$\Loss$ and asymptotically converges toward~$\elt_*$ because of the decaying learning rate~$\frac{\lambda_0}{k+1}$. In the meantime, non-smooth optimization methods that sample points around~$\elt_k$ in order to produce reliable descent directions converge in finite time. Namely, the classical Gradient Sampling method randomly samples~$3$ points and manages to reach an $(\epsilon,\eta)$-stationary point of~$\Loss$ in~$\sim 20.6 \pm 3.9$ iterations (averaged over 100 experiments), while our stratified approach improves on this by leveraging the fact that we explicitly have access to the two strata~$\{z_1<0\}$ and~$\{z_1>0\}$ where~$\Loss$ is differentiable. In particular, we only sample additional points when~$\elt_k$ is $\epsilon$-near the line $\{z_1=0\}$, and reach an $(\epsilon, \eta)$-stationary point in~18 iterations. Right plots showcase the evolution of the objective value~$f(\elt_k)$ and the distance to the minimum~$\|\elt_k - \elt_*\|$ across iterations. Parameters: $\elt_0 = (0.8, 0.8)$, $\lambda_0 = 10^{-1}$, $\epsilon = 10^{-1}$, $\eta = 10^{-2}$.}
    \label{fig:poc_intro}
\end{figure}

\subsection{Contributions and outline of contents}
\label{sec:outline_content}

Our new method, called {\em Stratified Gradient Sampling} (SGS), is a variation of the established GS algorithm, whose main steps for updating the current iterate~$\elt_k\in \R^n$ we recall in \cref{alg:gradient_sampling} below.
\begin{algorithm*}
\caption{An update step with the Gradient Sampling algorithm}
\label{alg:gradient_sampling}
\begin{algorithmic}[1]
\STATE{Sample $m\geq n+1$ points $\elt_{k}^1,\cdots, \elt_k^{m}$ in a ball $B(\elt_k,\epsilon)$}
\STATE{Compute {\em approximate subgradient} $G_k\defeq\{\nabla \Loss(\elt_k), \nabla \Loss(\elt_k^{1}), \cdots, \nabla \Loss(\elt_k^{m}) \}$}
\STATE{Compute {\em descent direction}~$g_k\defeq\argmin \{\|g\|^2, \, \, g \text{ in convex hull of }G_k\}$}
\STATE{Find {\em step size} $t_k\geq 0$ so that~$\Loss(\elt_k -t_k g_k)\leq \Loss(\elt_k) -\beta t_k \|g_k\|^2$ ($\beta \in (0,1)$ hyperparameter)}
\STATE{Ensure that~$\Loss$ is differentiable at~$\elt_{k+1}\defeq\elt_k -t_k g_k$ by small perturbations}
\end{algorithmic}
\end{algorithm*}
Our method is motivated by the closing remarks of a recent overview of the GS methodology \cite{burke2019gradient}, in which the authors suggest that the GS theory and practice could be enhanced by assuming some extra structure on top of the non differentiability of~$\Loss$. 

In this work, we deal with {\em stratifiably smooth maps}, for which the non differentiability is organized around smooth submanifolds that together form a stratification of~$\R^n$. In \cref{sec:background}, we review some background material in nonsmooth analysis and define stratifiably smooth maps~$f:\R^n\rightarrow \R$, a variant of the Whitney stratifiable functions from~\cite{bolte2007clarke} for which we do not impose any Whitney regularity on the gluing between adjacent strata of~$\R^n$, but rather enforce that there exist local $C^2$-extensions of the restrictions of~$f$ to top-dimensional strata.   

In order to update the current iterate~$\elt_k$ when minimizing a stratifiably smooth objective function~$\Loss$, we introduce a new descent direction~$g_k$. As in GS, $g_k$ is obtained in our new SGS algorithm by collecting the gradients of samples around~$\elt_k$ in an approximate subgradient~$G_k$, and then by taking the element with minimal norm in the convex set generated by~$G_k$. A key difference with GS is that we only need to sample as many points around~$\elt_k$ as there are distinct strata close by, compare with the~$m\geq n+1$ samples of \cref{alg:gradient_sampling}. In~\cref{prop:approx_descent}, we show that we indeed obtain a descent direction, i.e., that we have the descent criterion~$\Loss(\elt_k -t_k g_k)\leq \Loss(\elt_k) -\beta t_k \|g_k\|^2$ (as in Line~4 of \cref{alg:gradient_sampling}) for a suitable choice of step size~$t_k$.

Our SGS algorithm is detailed in \cref{sec:algorithm} and its analysis in \cref{sec:convergence}. The convergence of the original GS methodology crucially relies on the sample size~$m\geq n+1$ in order to apply the Carathéodory Theorem to subgradients. Differently, our convergence analysis relies on the fact that the gradients of~$\Loss$, when restricted to neighboring strata, are locally Lipschitz. Hence, our proof of asymptotic convergence to stationary points (\cref{thm:convergence_algo}) is substantially different. In \cref{thm:rate_convergence_algo}, we determine a convergence rate of our algorithm that holds for any proper stratifiably smooth map, which is an improvement over the guarantees of GS for general locally Lipschitz maps. Finally, in \cref{section_approx_distance_strata}, we adapt our method and results to the case where only estimated distances to nearby strata are available.

In \cref{sec:persistence_optim}, we introduce the persistence map $\persmap$ over a simplicial complex~$\SComplex$, which gives rise to a wide class of stratifiably smooth objectivefunctions with rich applications in TDA. We characterize strata around the current iterate (i.e., filter function)~$\elt_k$ by means of the permutation group over~$n$ elements, where~$n$ is the number of vertices in~$\SComplex$. Then, the Cayley graph associated to the permutation group allows us to use Dijkstra's algorithm to efficiently explore the set of neighboring strata by increasing order of distances to~$\elt_k$, that are needed to compute descent directions. 

\Cref{sec:experiments} is devoted to the implementation of the SGS algorithm for the optimization of persistent homology-based objective functions~$\Loss$. In \cref{subsec:total_pers}, we provide empirical evidence that SGS behaves better than GD and GS with a simple experiment about minimization of total persistence. In \cref{subsec:expe_topomean} and \cref{subsec:expe_registration}, we consider two novel topological optimization problems which we believe are of interest in real-world applications. On the one hand, the {\em Topological Template Registration} of a filter function~$\filt$ defined on a complex~$\SComplex$, is the task of finding a filter function~$\filt'$ over a smaller complex~$\SComplex'$ that preserves the barcode of~$\filt$. On the other hand, given a Mapper graph~$G$, which is a standard visualization tool for arbitrary data sets~\cite{Singh2007}, we can bootstrap the data set in order to produce multiple bootstrapped graphs~$G_i$. The {\em Topological Mean} is then the task of finding a new graph $G^*$ whose barcode is as close as possible to the mean of the barcodes associated to the graphs~$G_i$. As a result we obtain a smoothed version $G^*$ of the Mapper graph~$G$ in which spurious and non-relevant graph attributes are removed. 

\section{A direction of descent for stratifiably smooth maps}
\label{sec:background}
In this section, we define the class of stratifiably smooth maps whose optimization is at stake in this work. For such maps, we can define an approximate subgradient and a corresponding descent direction, which is the key ingredient of our algorithm. 

\subsection{Nonsmooth Analysis}
\label{sec:prelim_non_smooth_analysis}
We first recall some useful background in nonsmooth analysis, essentially from~\cite{clarke1990optimization}. Throughout, $f:\R^n\rightarrow\R$ is a locally Lipschitz (non necessary smooth, nor convex) and proper (i.e., compact sublevel sets) function, which we aim to minimize.

First-order information of $f$ at $\elt \in\R^n$ in a direction  $v \in \R^n$  is captured by its {\em generalized directional derivative}:
\begin{equation}
	f^\circ(\elt; v) = \limsup_{ y \to x, t \shortdownarrow 0^+} \frac{f(y + tv) - f(y)}{t},
\end{equation}
Besides, the \emph{generalized gradient} is the following set of linear subapproximations:
\begin{equation}
\label{eq:generalized_gradient}
	\partial f (\elt)\defeq \big\{ \zeta \in \R^n,\ f^\circ(\elt; v) \geq \braket{\zeta, v} \text{ for all } v \in \R^n\big\}.	
\end{equation}
Given an arbitrary set~$S\subset \R^n$ of Lebesgue measure~$0$, we have an alternative description of the generalized gradient in terms of limits of surrounding gradients, whenever defined:
\begin{equation}
\label{eq:generalized_gradient_limiting_gradients}
	\partial f (\elt)= \co \big\{\lim \nabla f (\elt_i) \, | \, \elt_i\rightarrow \elt, \, \nabla f (\elt_i) \text{ is defined }, \, \lim \nabla f (\elt_i) \text{ exists}, \, \elt_i\notin S \big\},
\end{equation}
where~$\co$ is the operation of taking the closure of the convex hull.\footnote{Here the equality holds as well (by some compactness argument) when taking the convex hull without closure. As we take closed convex hulls later on, we choose not to introduce this subtlety explicitly.} The duality between generalized directional derivatives and gradients is captured by the equality:
\begin{equation}
	f^\circ(\elt; v) = \max\{ \braket{\zeta, v}, \zeta \in \partial f(\elt) \}.
\end{equation}
The {\em Goldstein subgradient}~\cite{goldstein1977optimization} is an $\epsilon$-relaxation of the generalized gradient:
\begin{equation}
\label{eq:goldstein_subgradient}
	\Goldstein (\elt)= \co \big\{ \lim \nabla f (\elt_i) \, | \, \elt_i\rightarrow \elt', \, \nabla f (\elt_i) \text{ is defined }, \, \lim \nabla f (\elt_i) \text{ exists, } \, |x-x'|\leqslant \epsilon  \big\}.
\end{equation}
Given~$\elt\in \R^n$, we say that:
\vspace{1mm}
\begin{center}
$\elt$ is a {\em stationary point} (for $f$) if $0\in \partial f (\elt)$.
\end{center}
\vspace{1mm}
Any local minimum is stationary, and conversely if $f$ is convex. We also have weaker notions. Namely, given~$\epsilon, \eta>0$, 
\begin{center}
$\elt$ is {\em $\epsilon$-stationary} if~$0\in \Goldstein (\elt)$; and\\ 
$\elt$ is {\em $(\epsilon,\eta)$-stationary} if $d(0,\Goldstein (\elt))\leqslant \eta$.
\end{center}
\vspace{2mm}
\subsection{Stratifiably smooth maps}
\label{subsec:stratifiably_smooth}
Desirable properties for an optimization algorithm is that it produces iterates~$(\elt_k)_k$ that either converge to an $(\epsilon,\eta)$-stationary point in finitely many steps, or whose subsequences (or some of them) converge to an $\epsilon$-stationary point. For this, we work in the setting of objective functions that are smooth when restricted to submanifolds, that together partition~$\R^n$.
\begin{definition}
\label{def:stratification}
A {\em stratification} $\Stratif=\{\Stratum_i\}_{i\in I}$ of a closed subset $\mathbb{X}\subseteq \R^n$ is a locally finite partition of~$\mathbb{X}$ by smooth submanifolds $\Stratum_i$---called {\em strata}---such that for $i\neq j\in I$:
\[\overbar{\Stratum_i} \cap \Stratum_j \neq \emptyset \Rightarrow \Stratum_j \subseteq \overbar{\Stratum_i}\setminus \Stratum_i .\]
This makes $(\mathbb{X},\Stratif)$ into a {\em stratified space}.
\end{definition}
Note that we do not impose any (usually needed) gluing conditions between adjacent strata, as we do not require them in the analysis. In particular, semi-algebraic, subanalytic, or definable subsets of~$\R^n$, together with Whitney stratified sets are stratified in the above weak sense. We next define the class of maps~$\Loss$ with smooth restrictions~$f_{|\Stratum_i}$ to strata~$\Stratum_i$ of some stratification~$\Stratif$, inspired by the Whitney stratifiable maps of~\cite{bolte2007clarke} (there~$\Stratif$ is required to be Whitney) and the \emph{stratifiable functions} of~\cite{drusvyatskiy2015clarke}, however we further require that the restrictions~$f_{|\Stratum_i}$ admit local extensions of class~$C^2$. 
\begin{definition}
\label{def:stratified_map}
The map $\Loss:\R^n\rightarrow \R$ is {\em stratifiably smooth} if there exists a stratification $\Stratif$ of $\R^n$, such that for each stratum $\Stratum_i\in \Stratif$, the restriction~$\Loss_{|\Stratum_{i}}$ admits an extension~$\Loss_{i}$ of class~$C^2$ in a neighborhood of $\Stratum_i$.
\end{definition}
\begin{remark}
\label{remark_stratified_map_def}
The slightly weaker assumption that the extension~$\Loss_{i}$ is continuously differentiable with locally Lipschitz gradient would have also been sufficient for our purpose. 
\end{remark} 

We denote by~$\Stratum_\elt$ the stratum containing $\elt$, and by $\Stratif_x\subseteq \Stratif$ the set of strata containing~$\elt$ in their closures. More generally, for $\epsilon> 0$, we let $\Stratif_{\elt,\epsilon}\subseteq \Stratif$ be the set of strata $\Stratum_i$ such that the closure of the  ball $B(\elt,\epsilon)$ has non-empty intersection with the closure of $\Stratum_i$. Local finiteness in the definition of a stratification implies that $\Stratif_{\elt,\epsilon}$ (and $\Stratif_{\elt}$) is finite. 

If~$f$ is stratifiably smooth and $\Stratum_i\in \Stratif_\elt$ is a stratum, there is a well-defined limit gradient~$\nabla_{\Stratum_{i}}f(\elt)$, which is the unique limit of the gradients $\nabla f_{|{\Stratum_i}}(\elt_n)$ where $\elt_n\in \Stratum_i$ is any sequence converging to $\elt$. Indeed, this limit exists and does not depend on the choice of sequence since $f_{|\Stratum_{i}}$ admits a local $C^2$ extension $f_i$. The following result states that the generalized gradient at $\elt$ can be retrieved from these finitely many limit gradients along the various adjacent top-dimensional strata. 
\begin{proposition}
\label{prop:generalized_gradient_along_strata}
If~$f$ is stratifiably smooth, then for any $\elt\in \R^n$ we have:
\[\partial f(\elt)= \mathrm{co} \big\{ \nabla_{\Stratum_{i}}f(\elt), \text{ $\Stratum_i\in \Stratif_\elt$ is of dimension $n$}  \big\}.\]
More generally, for $\epsilon>0$:
\[\Goldstein(\elt)= \co \big\{ \nabla_{\Stratum_{i}}f(\elt') \,| \, |x'-x|\leqslant \epsilon, \, \Stratum_i\in \Stratif_{\elt'}\subseteq \Stratif_{\elt,\epsilon}\text{ is of dimension $n$}  \big\}.\]
\end{proposition}
\begin{proof}
We show the first equality only, as the second can be proven along the same lines. We use the description of~$\partial f(\elt)$ in terms of limit gradients from \Cref{eq:generalized_gradient_limiting_gradients}, which implies the inclusion of the right-hand side in $\partial f(\elt)$. Conversely, let~$S$ be the union of strata in $\Stratif_\elt$ with positive codimension, which is of measure~$0$. Let $\elt_i$ be a sequence avoiding~$S$, converging to $\elt$, such that $\nabla f(\elt_i)$ converges as well. Since $\Stratif_\elt$ is finite, up to extracting a subsequence, we can assume that all $\elt_i$ are in the same top-dimensional stratum $\Stratum_i\in \Stratif_\elt$. Consequently,~$\nabla f(\elt_i)$ converges to~$\nabla_{\Stratum_{i}}f(\elt)$.
\end{proof}
\subsection{Direction of descent}
\label{sec:approx_goldstein_descent_direction}

Thinking of~$\elt$ as a current position, we look for a direction of {\em (steepest) descent}, in the sense that a perturbation of~$\elt$ in this direction produces a (maximal) decrease of~$\Loss$. Given~$\epsilon \geqslant 0$, we let~$g(\elt,\epsilon)$ be the projection of the origin on the convex set~$\Goldstein (\elt)$. Equivalently,~$g(\elt,\epsilon)$ solves the minimization problem:
\begin{equation}
\label{eq:g_minimal_norm}
g(\elt,\epsilon) = \argmin \big\{\|g\|, \, g\in \Goldstein (\elt) \big\}.
\end{equation}
Introduced in~\cite{goldstein1977optimization}, the direction $-g(\elt,\epsilon)$ is a good candidate of direction for descent, as we explain now. Since $g(\elt,\epsilon)$ is the projection of the origin on the convex closed set $\Goldstein (\elt)$, we have the classical inequality~$\braket{g(\elt,\epsilon), g(\elt,\epsilon)-g} \leqslant 0$ that holds for any~$g$ in the Goldstein subgradient at~$\elt$. Equivalently, 
\begin{equation}
\label{eq:g_dot_product}
 \forall g\in \Goldstein (\elt), \, \braket{-g(\elt,\epsilon), g} \leqslant -\|g(\elt,\epsilon)\|^2.
\end{equation}
Informally, if we think of a small perturbation $\elt-tg(\elt,\epsilon)$ of $\elt$ along this direction, for~$t>0$ small enough, then $\Loss(\elt-tg(\elt,\epsilon)) \approx \Loss(\elt)-t \braket{\nabla \Loss(\elt),g(\elt,\epsilon)}$. Using \Cref{eq:g_dot_product}, since $\nabla \Loss(\elt)\in \Goldstein (\elt)$, 
we deduce that $\Loss(\elt-tg(\elt,\epsilon))\leqslant \Loss(\elt)- t \|g(\elt,\epsilon)\|^2$. So~$\Loss$ locally decreases at linear rate in the direction $-g(\elt,\epsilon)$. This intuition relies on the fact that $\nabla \Loss(\elt)$ is well-defined so as to provide a first order approximation of~$\Loss$ around~$\elt$, and that $t$ is chosen small enough. In order to make this reasoning completely rigorous, we need the following well-known result (see for instance~\cite{clarke1990optimization}):

\begin{theorem}[Lebourg Mean value Theorem]
\label{prop:lebourg}
Let $\elt, \elt'\in \R^n$. Then there exists some $y\in [\elt,\elt']$ and some $w\in \partial \Loss(y)$ such that:
\[\Loss(\elt')-\Loss(\elt)=\braket{w, \elt'-\elt}.\]
\end{theorem}
Let $t>0$ be lesser than $\frac{\epsilon}{\|g(\elt,\epsilon)\|}$, in order to ensure that $\elt'\defeq \elt - t g(\elt,\epsilon)$ is $\epsilon$-close to $\elt$. Then by the mean value theorem (and \cref{prop:generalized_gradient_along_strata}), we have that 
\[\Loss(\elt-tg(\elt,\epsilon))- \Loss(\elt)= -t \braket{w, g(\elt,\epsilon)} \]
for some $w\in \Goldstein(\elt)$. \Cref{eq:g_dot_product} yields
\begin{equation}
\label{eq:descent}
\forall t\leqslant \frac{\epsilon}{\|g(\elt,\epsilon)\|}, \qquad \Loss(\elt-tg(\elt,\epsilon))\leqslant  \Loss(\elt)-t\|g(\elt,\epsilon)\|^2,
\end{equation}
as desired.

In practical scenarios however, it is unlikely that the exact descent direction $-g(\elt,\epsilon)$ could be determined. Indeed, from \cref{eq:g_minimal_norm}, it would require the knowledge of the set $\Goldstein(\elt)$, which consists of infinitely many (limits of) gradients in an $\epsilon$-neighborhood of $\elt$. We now build, provided $\Loss$ is stratifiably smooth, a faithful approximation $\AppGoldstein(\elt)$ of $\Goldstein(\elt)$, by collecting gradient information in the strata that are $\epsilon$-close to~$\elt$.

For each top-dimensional stratum $\Stratum_i\in \Stratif_{\elt,\epsilon}$, let $\elt_i$ be an arbitrary point in $\overbar{\Stratum}_i\cap \overbar{B}(\elt,\epsilon)$. Define
\begin{equation}
\label{eq:approximate_goldstein}
\AppGoldstein(\elt)\defeq\co \big\{ \nabla_{\Stratum_{i}} \Loss(\elt_i) , \,   \Stratum_i\in \Stratif_{\elt,\epsilon}  \big\}.
\end{equation}
Of course, $\AppGoldstein(\elt)$ depends on the choice of each~$\elt_i\in  \Stratum_i$. But this will not matter for the rest of the analysis, as we will only rely on the following approximation result which holds for arbitrary choices of points~$\elt_i$:
\begin{proposition}
\label{prop:approximation_goldstein_generalized_gradient}
Let $\elt\in \R^n$ and $\epsilon>0$. Assume that $\Loss$ is stratifiably smooth. Let $\lipconst$ be a Lipschitz constant of the gradients $\nabla f_{i}$ restricted to $\overbar{B}(\elt,\epsilon)\cap \Stratum_i$, where $\Loss_i$ is some local $C^2$ extension of $\Loss_{|\Stratum_i}$, and $ \Stratum_i\in \Stratif_{\elt,\epsilon} $ is top dimensional. Then we have:
\[\AppGoldstein(\elt) \subseteq \Goldstein(\elt) \subseteq \AppGoldstein(\elt)+  \overbar{B}(0, 2\lipconst \epsilon). \]
In particular, $d_H(\AppGoldstein(\elt), \Goldstein(\elt))\leqslant 2\lipconst \epsilon$.
\end{proposition}
Note that, since the $\Loss_i$ are of class $C^2$, their gradients are locally Lipschitz, hence by compactness of $\overbar{B}(\elt,\epsilon)$, the existence of the Lipschitz constant $\lipconst$ above is always guaranteed.
\begin{proof}
From \cref{prop:generalized_gradient_along_strata}, we have
\[\Goldstein(\elt)= \co \big\{ \nabla_{\Stratum}\Loss(\elt') \,| \, |\elt'-\elt|\leqslant \epsilon, \, \Stratum\in \Stratif_{\elt'}\subseteq \Stratif_{\elt,\epsilon}\text{ is of dimension $n$ }  \big\}.\]
This yields the inclusion $\AppGoldstein(\elt) \subseteq \Goldstein(\elt)$. Now, let $\elt'\in \R^n$, $|\elt'-\elt|\leqslant \epsilon$, and let $\Stratum_i \in \Stratif_{\elt'}\subseteq \Stratif_{\elt,\epsilon}$ be a top-dimensional stratum touching $\elt'$. 
Based on how $\elt_i$ is defined in \cref{eq:approximate_goldstein}, we have that~$\elt'$ and $\elt_i$ both belong to $\overbar{B}(\elt,\epsilon)$, and they both belong to the stratum $\Stratum_i$. Therefore, $|\nabla_{\Stratum_i}\Loss(\elt')-\nabla_{\Stratum_i}\Loss(\elt_i)|\leqslant 2\lipconst  \epsilon$, and so $\nabla_{\Stratum_i}\Loss(\elt')\in \AppGoldstein(\elt)+  \overbar{B}(0, 2\lipconst \epsilon)$. The result follows from the fact that $\AppGoldstein(\elt)+  \overbar{B}(0, 2\lipconst \epsilon)$ is convex and closed.
\end{proof}
Recall from \cref{eq:descent} that the (opposite to the) descent direction $-g(\elt,\epsilon)$ is built as the projection of the origin onto $\Goldstein(\elt)$. Similarly, we define our approximate descent direction as $-\tilde{g}(\elt,\epsilon)$, where $\tilde{g}(\elt,\epsilon)$ is the projection of the origin onto the convex closed set $\AppGoldstein(\elt)$:
\begin{equation}
\label{eq:approx_g_minimal_norm}
\tilde{g}(\elt,\epsilon) = \argmin \big\{\|\tilde{g}\|, \, \tilde{g}\in \AppGoldstein (\elt) \big\}.
\end{equation}
We show that this choice yields a direction of decrease of $\Loss$, in a sense similar to \cref{eq:descent}.
\begin{proposition}
\label{prop:approx_descent}
Let~$\Loss$ be stratifiably smooth, and let~$\elt$ be a non-stationary point. Let $0<\beta<1$, and $\epsilon_0>0$. Denote by~$\lipconst$ a Lipschitz constant for all gradients of the restrictions~$\Loss_i$ to the ball $\overbar{B}(\elt,\epsilon_0)$ (as in \cref{prop:approximation_goldstein_generalized_gradient}). Then:
\begin{itemize}
\item[{\bf (i)}] For $0<\epsilon\leqslant \epsilon_0$ small enough we have~$\epsilon\leqslant \frac{1-\beta}{2\lipconst}\|\tilde{g}(\elt,\epsilon)\|$; and 
\item[{\bf (ii)}] For such~$\epsilon$, we have~$\forall t \leqslant \frac{\epsilon}{\|\tilde{g}(\elt,\epsilon)\|}, \, \Loss(\elt-t\tilde{g}(\elt,\epsilon))\leqslant \Loss(\elt)-\beta t\|\tilde{g}(\elt,\epsilon)\|^2$.
\end{itemize} 
\end{proposition}
\begin{proof}
Saying that~$\elt$ is non-stationary is equivalent to the inequality $\|g(\elt,0)\|>0$. We show that the map $\epsilon \in \R^+ \mapsto \|g(\elt,\epsilon)\|\in \R^+$, which is non-increasing, is continuous at~$0^+$. Let~$\epsilon$ be small enough such that the sets of strata incident to~$\elt$ are the same that meet with the~$\epsilon$-ball around~$\elt$, i.e., $\Stratif_{\elt,\epsilon}=\Stratif_{\elt}$. Such an~$\epsilon$ exists since there are finitely many strata, which are closed sets, that meet with a sufficiently small neighborhood of~$\elt$. Of course, all smaller values of~$\epsilon$ enjoy the same property. By \cref{prop:generalized_gradient_along_strata}, we then have the nesting
\[\partial f (\elt) \subseteq \Goldstein(\elt) \subseteq \partial f (\elt) +\overbar{B}(0, 2\lipconst \epsilon),\]
where $\lipconst$ is a Lipschitz constant for the gradients in neighboring strata. In turn,~$0\leqslant \|g(\elt,0)\|-\|g(\elt,\epsilon)\|\leqslant 2\lipconst \epsilon$. In particular,~$\|g(\elt,\epsilon)\|\rightarrow \|g(\elt,0)\|>0$ as~$\epsilon$ goes to~$0$, hence~$\epsilon=o(\|g(\elt,\epsilon)\|)$. Besides, the inclusion $\AppGoldstein(\elt)\subseteq \Goldstein(\elt)$ (\cref{prop:approximation_goldstein_generalized_gradient}) implies that~$\|\tilde{g}(\elt,\epsilon)\|\geqslant \|g(\elt,\epsilon)\|>0$. This yields $\epsilon=o(\|\tilde{g}(\elt,\epsilon)\|)$ and so item~{\bf (i)} is proved.

We now assume that~$\epsilon$ satisfies the inequality of item~{\bf (i)}, and let~$0\leqslant t \leqslant \frac{\epsilon}{\|\tilde{g}(\elt,\epsilon)\|}$. By the Lebourg mean value theorem, there exists a $y\in [\elt,\elt-t\tilde{g}(\elt,\epsilon)]$ and some~$w\in \partial \Loss(y)$ such that: 
\[\Loss(\elt-t\tilde{g}(\elt,\epsilon))- \Loss(\elt)=t \braket{w,-\tilde{g}(\elt,\epsilon)}. \]
Since $t\leqslant \frac{\epsilon}{\|\tilde{g}(\elt,\epsilon)\|}$,~$y$ is at distance no greater than $\epsilon$ from $\elt$. In particular, $w$ belongs to $\Goldstein(\elt)$. From \cref{prop:approximation_goldstein_generalized_gradient}, there exists some $\tilde{w}\in \AppGoldstein(\elt)$ at distance no greater than $2\lipconst \epsilon$ from $w$. We then rewrite:
\begin{equation}
\label{eq:proposition_approx_descent_1}
\Loss(\elt-t\tilde{g}(\elt,\epsilon))- \Loss(\elt)=t \braket{w-\tilde{w},-\tilde{g}(\elt,\epsilon)} + t \braket{\tilde{w},-\tilde{g}(\elt,\epsilon)}. 
\end{equation}
On the one hand, by the Cauchy-Schwarz inequality:
\begin{equation}
\label{eq:proposition_approx_descent_2}
 \braket{w-\tilde{w},-\tilde{g}(\elt,\epsilon)} \leqslant |w-\tilde{w}|\cdot\|\tilde{g}(\elt,\epsilon)\|\leqslant 2 \lipconst \epsilon \|\tilde{g}(\elt,\epsilon)\|\leqslant  (1-\beta) \|\tilde{g}(\elt,\epsilon)\|^2,
\end{equation}
where the last inequality relies on the assumption that $\epsilon\leqslant \frac{1-\beta}{2\lipconst}\|\tilde{g}(\elt,\epsilon)\|$. On the other hand, since $\tilde{g}(\elt,\epsilon)$ is the projection of the origin onto $\AppGoldstein(\elt)$, we obtain $\braket{\tilde{g}(\elt,\epsilon)- \tilde{w}, \tilde{g}(\elt,\epsilon)} \leqslant 0$, or equivalently:
\begin{equation}
\label{eq:proposition_approx_descent_3}
\braket{\tilde{w},-\tilde{g}(\elt,\epsilon)} \leqslant -\|\tilde{g}(\elt,\epsilon)\|^2.
\end{equation}
Plugging the inequalities of \Cref{eq:proposition_approx_descent_2,eq:proposition_approx_descent_3} into \cref{eq:proposition_approx_descent_1} proves item {\bf (ii)}.
\end{proof}
\section{Stratified Gradient Sampling (SGS)}
\label{sec:algo_stratified_descent}
In this section we develop a gradient descent algorithm for the optimization of stratifiably smooth functions, and then we detail its convergence properties. We require  that the objective function~$\Loss:\R^n\rightarrow \R$ has the following properties:
\begin{itemize}
    \item {\bf(Proper)}: {\em $\Loss$ has compact sublevel sets.}
    \item {\bf(Stratifiably smooth)}: {\em $\Loss$ is stratifiably smooth, and for each iterate~$\elt$ and~$\epsilon\geq0$ we have an oracle~$\SampleOracle(\elt,\epsilon)$ that samples one $\epsilon$-close element~$\elt'$ in each $\epsilon$-close top-dimensional stratum~$\Stratum'$.}
    \item{\bf (Differentiability check)} {\em We have an oracle~$\DiffOracle(\elt)$ checking whether an iterate~$\elt\in \R^n$ belongs to the set~$\DiffSet \subset \R^n$ over which~$\Loss$ is differentiable.}
\end{itemize}
That~$\Loss$ is a proper map is also needed in the original GS algorithm~\cite{burke2005robust}, but is a condition that can be omitted as in~\cite{kiwiel2007convergence} to allow the values~$\Loss(\elt_k)$ to decrease to~$-\infty$. In our case we stick to this assumption because we need the gradient of~$\Loss$ (whenever defined) to be Lipschitz on sublevel sets. 

Similarly, the ability to check that an iterate~$\elt_k$ belongs to~$\DiffSet$ is standard in the GS methodology. We use it to make sure that~$\Loss$ is differentiable at each iterate~$\elt_k$. For this, we call a subroutine~$\AlgoMakeDiff$ which slightly perturbs the iterate~$\elt_k$ to achieve differentiability and to maintain a descent condition. Note that these considerations are mainly theoretical because generically the iterates~$\elt_k$ are points of differentiability, hence~$\AlgoMakeDiff$ is unlikely to change anything. 

The last requirement that~$\Loss$ is stratifiably smooth replaces the classical weaker assumption used in the GS algorithm that~$\Loss$ is locally Lipschitz and that the set~$\DiffSet$ where~$\Loss$ is differentiable is open and dense. There are many possible ways to design the oracle~$\SampleOracle(\elt,\epsilon)$: for instance the sampling could depend upon arbitrary probability measures on each stratum, or it could be set by deterministic rules depending on the input~$(\elt,\epsilon)$ as will be the case for the persistence map in \cref{sec:persistence_optim}. However our algorithm and its convergence properties are oblivious to these degrees of freedom, as by \cref{sec:approx_goldstein_descent_direction} any sampling allows us to approximate the Goldstein subgradient~$\Goldstein(\elt_k)$ using finitely many neighbouring points to compute~$\AppGoldstein(\elt_k)$. In turn we have an approximate descent direction~$g_k$ which can be used to produce the subsequent iterate~$\elt_{k+1}\defeq\elt_k-t_kg_k$ as in the classical smooth gradient descent.
\subsection{The algorithm}
\label{sec:algorithm}
The details of the main algorithm~$\AlgoGradientDescent$ are given in \cref{alg:minimization}.

The algorithm~$\AlgoGradientDescent$ calls the method~$\AlgoUpdateStep$ of~\cref{alg:update step} as a subroutine to compute the right descent direction~$g_k$ and the right step size~$t_k$. Essentially, this method progressively reduces the exploration radius~$\epsilon_k$ of the ball where we compute the descent direction~$g_k\defeq\tilde{g}(\elt_k,\epsilon_k)$ until the criteria of \cref{prop:approx_descent} ensuring that the loss sufficiently decreases along~$g_k$ are met.

Given the iterate~$\elt_k$ and the radius~$\epsilon_k$, the calculation of~$g_k\defeq\tilde{g}(\elt_k,\epsilon_k)$ is done by the subroutine $\AlgoGeneralizedGradient$ in \cref{alg:generalized gradient}: points~$\elt'$ in neighboring strata that intersect the ball~$B(\elt_k,\epsilon_k)$ are sampled using~$\SampleOracle(\elt_k,\epsilon_k)$ to compute the approximate Goldstein gradient and in turn the descent direction~$g_k$. 

Much like the classical GS algorithm, our method 
behaves like the well-known smooth gradient descent where the gradient is replaced with a descent direction computed from gradients in neighboring strata. A key difference however is that, in order to find the right exploration radius~$\epsilon_k$ and step size~$t_k$, the $\AlgoUpdateStep$ needs to maintain a constant~$\ControlConst_k$ to approximate the ratio~$\frac{1-\beta}{2\lipconst}$ of \cref{prop:approx_descent}, as no Lipschitz constant~$\lipconst$ may be explicitly available.

To this effect, $\AlgoUpdateStep$ maintains a relative balance between the exploration radius~$\epsilon_k$ and the norm of the descent direction~$g_k$, controlled by~$\ControlConst_k$, i.e.,~$\epsilon_k \simeq \ControlConst_k \|g_k\|$. As we further maintain~$\ControlConst_k \simeq \frac{1-\beta}{2\lipconst}$, we know that the convergence properties of~$\epsilon_k$ and~$g_k$ are closely related. Thus, the utility of this controlling constant is mainly theoretical, to ensure convergence of the iterates~$\elt_k$ towards stationary points in \cref{thm:convergence_algo}. In practice, we start with a large initial constant~$\ControlConst_0$, and decrease it on line 11 of \cref{alg:update step} whenever it violates a property of the target constant~$\frac{1-\beta}{2\lipconst}$ given by \cref{prop:approx_descent}. 
\begin{algorithm}
\caption{$\AlgoGradientDescent(\Loss,\elt_0,\epsilon,\eta,\ControlConst_0,\beta,\gamma)$}
\label{alg:minimization}
\begin{algorithmic}[1]
\REQUIRE{Loss function~$\Loss$, initial iterate~$\elt_0\in \DiffSet$, exploration radius $\epsilon>0$, initial constant~$\ControlConst_0>0$ controlling exploration radius, critical distance to origin $\eta\geqslant 0$, descent rate $0<\beta<1$, step size decay rate~$0<\gamma<1$}
\STATE {$k\gets 0$}
\REPEAT
\STATE{$(t_k,g_k,\ControlConst_{k+1})\gets \AlgoUpdateStep(\Loss,\elt_k,\epsilon,\eta,\ControlConst_k, \beta, \gamma)$ via \cref{alg:update step}}
\STATE{$\elt_{k+1}\gets \elt_k-t_k g_k$}
\STATE{$\elt_{k+1}\gets \AlgoMakeDiff(\elt_{k+1},\elt_{k},t_k,g_k)$}
\STATE{$k \gets k+1$}
\UNTIL{$\|g_k\| \leq \eta$}
\RETURN{$\elt_k$}
\end{algorithmic}
\end{algorithm}

\begin{algorithm}
\caption{$\AlgoUpdateStep(\Loss,\elt_k,\epsilon,\eta,\ControlConst_k, \beta, \gamma)$}
\label{alg:update step}
\begin{algorithmic}[1]
%
%
\STATE {$\epsilon_k \gets \epsilon$ and~$\ControlConst_{k+1}\gets \ControlConst_{k}$}
\REPEAT
\STATE{$g_k\gets \AlgoGeneralizedGradient(\elt_k,\epsilon_k)$ via \cref{alg:generalized gradient} }
\IF{$\|g_k\|\leqslant \eta$ } 
\STATE{Break, \textbf{return} $t_k=0$, $g_k$ and~$\ControlConst_{k+1}$  {\em \hspace{1.8cm}  \small{\# Set $\eta=0$ to reach an~$\epsilon$-stationary point}}}
\ENDIF
\STATE{$t_k \gets \frac{\epsilon_k}{\|g_k\|}$}     {\em \hspace{8.5cm} \small{\# Candidate of update step}}
\WHILE{$\Loss(\elt_k-t_k g_k) > \Loss(\elt_k)-\beta t_k\|g_k\|^2$ and $\epsilon_k\leqslant \ControlConst_{k+1}\|g_k\|$} 
\STATE{$\ControlConst_{k+1} \gets \gamma \ControlConst_{k+1}$} {\em \hspace{1cm} \small{\# Once $\ControlConst_{k+1}\leq \frac{1-\beta}{2\lipconst}$, this loop never occurs by (ii) of~\cref{prop:approx_descent}}}
\ENDWHILE
\IF{$\Loss(\elt_k-t_kg_k)> \Loss(\elt_k)-\beta t_k\|g_k\|^2$ or $\epsilon_k> \ControlConst_{k+1} \|g_k\|$}
\STATE{$\epsilon_k \gets \gamma \epsilon_k$} {\em \hspace{4.3cm} \small{\# Reduce~$\epsilon_k$ to satisfy criterion (i) of \cref{prop:approx_descent}}}
\ENDIF
\UNTIL{$\Loss(\elt_k-t_kg_k)< \Loss(\elt_k)-\beta t_k\|g_k\|^2$ and $\epsilon_k< \ControlConst_{k+1} \|g_k\|$}
\RETURN{$t_k$, $g_k$ and~$\ControlConst_{k+1}$ }

\end{algorithmic}
\end{algorithm}

\begin{algorithm}
\caption{$\AlgoGeneralizedGradient(\elt_k,\epsilon_k)$}
\label{alg:generalized gradient}
\begin{algorithmic}[1]
\STATE {$G_k\gets \{\nabla \Loss(\elt_k)\}$} {\em \hspace{1cm} \small{\# Eventually~$G_k$ will be some approximate Goldstein subgradient $\tilde{\partial} f_{\epsilon_k}(\elt_k)$}}

\STATE{$\{\elt_k^1,\cdots,\elt_k^{m}\}\gets \SampleOracle(\elt_k,\epsilon_k)$}  {\em  \hspace{1cm}\small \# $\epsilon$-close samples from $\epsilon$-close top dim strata}
\FOR{$1\leq l \leq m$} 
%
%
\STATE{$G_k\gets G_k \cup \{\nabla \Loss(\elt_k^l)\}$ } {\em \hspace{5cm} \small{\# Add gradients from remote strata }}
\ENDFOR
\STATE{Solve the quadratic minimization problem~$g_k=\argmin \{\|g\|^2, \, \, g\in \co(G_k)\}$}
\RETURN{$g_k$} {\em \hspace{3cm} \small{\# $g_k=\tilde{g}(\elt_k,\epsilon_k)$ is the approximate steepest descent direction}}
\end{algorithmic}
\end{algorithm}
\begin{algorithm}
\caption{$\AlgoMakeDiff(\elt_{k+1},\elt_{k},t_k,g_k)$}
\label{alg:make_differentiable}
\begin{algorithmic}[1]
\STATE{$r\gets t_k \|g_k\|$}
\WHILE{$\elt_{k+1}\notin \DiffSet$ or $\Loss(\elt_{k+1})> \Loss(\elt_k)-\beta t_k\|g_k\|^2$}
\STATE{Replace $\elt_{k+1}$ with a sample in $B(\elt_k-t_kg_k,r)$}
\STATE{$r \gets \frac{r}{2}$}
\ENDWHILE
\RETURN{$\elt_{k+1}$}
\end{algorithmic}
\end{algorithm}
\begin{remark}
\label{rke:update_lipschitz_known}
Assume that we dispose of a common Lipschitz constant~$\lipconst$ for all gradients~$\nabla \Loss_i$ in the~$\epsilon$-neighborhood of the current iterate~$\elt_k$, recall that~$\Loss_i$ is any~$C^2$ extension of the restriction~$f_{|\Stratum_i}$ to the neighboring top-dimensional stratum~$\Stratum_i\in \Stratif_{\elt,\epsilon}$. Then we can simplify \cref{alg:update step} by decreasing the exploration radius~$\epsilon_k$ progressively until~$\epsilon_k\leq \frac{(1-\beta)}{2\lipconst} \|\tilde{g}(\elt_k,\epsilon_k)\|$ as done in \cref{alg:update step with lipschitz}: This ensures by \cref{prop:approx_descent} that the resulting update step satisfies the descent criterion~$\Loss(\elt_k-t_k\tilde{g}(\elt_k,\epsilon_k)) < \Loss(\elt_k)-\beta t_k\|\tilde{g}(\elt_k,\epsilon_k)\|^2$. In particular the parameter~$\ControlConst_k$ is no longer needed, and the theoretical guarantees of the simplified algorithm are unchanged. Note that for objective functions from TDA (see \cref{sec:persistence_optim}), the stability theorems (e.g. from~\cite{cohen2007stability}) often provide global Lipschitz constants, enabling the use of the simplified update step described in \cref{alg:update step with lipschitz}.  
\begin{algorithm}
\caption{$\AlgoUpdateStepSimple(\Loss,\elt_k,\epsilon,\eta, \beta,\gamma)$}
\label{alg:update step with lipschitz}
\begin{algorithmic}[1]
\STATE {$\epsilon_k \gets \epsilon$ and $g_k\gets \AlgoGeneralizedGradient(\elt_k,\epsilon_k)$ via \cref{alg:generalized gradient}}
\REPEAT
\IF{$\|g_k\|\leqslant \eta$ } 
\STATE{Break, \textbf{return} $t_k=0$ and $g_k$  {\em \hspace{1.8cm}  \small{\# Set $\eta=0$ to reach an~$\epsilon$-stationary point}}}
\ENDIF
\STATE {$\epsilon_k \gets \gamma \epsilon_k$ and $g_k\gets \AlgoGeneralizedGradient(\elt_k,\epsilon_k)$ via \cref{alg:generalized gradient}}
\UNTIL{$\epsilon_k\leq \frac{1-\beta}{2\lipconst} \|g_k\|$}
\RETURN{$t_k\defeq \frac{\epsilon_k}{\|g_k\|}$ and $g_k$ }
\end{algorithmic}
\end{algorithm}
\end{remark}

\begin{remark}
\label{rke:quick_update}
In the situation of \cref{rke:update_lipschitz_known}, let us further assume that~$\epsilon \in \R_+ \mapsto \|\tilde{g}(\elt,\epsilon)\|\in \R_+$ is non-increasing. 
This monotonicity property mimics the fact that~$\epsilon \in \R_+ \mapsto \|g(\elt,\epsilon)\| \in \R_+$ is non-increasing, since increasing~$\epsilon$ grows the Goldstein generalized gradient~$\Goldstein(\elt)$, of which~$g(\elt,\epsilon)$ is the element with minimal norm. 
If the initial exploration radius~$\epsilon$ does not satisfy the termination criterion (Line~8),~$\epsilon \leq \frac{1-\beta}{2\lipconst} \|\tilde{g}(x_k, \epsilon)\|$, then one can set~$\epsilon_k \defeq \frac{1-\beta}{2\lipconst} \|\tilde{g}(\elt_k,\epsilon)\| \leq \epsilon$, yielding $\epsilon_k \leq \frac{1-\beta}{2\lipconst} \|\tilde{g}(\elt_k,\epsilon_k)\|$. 
This way \cref{alg:update step with lipschitz} is further simplified: in constant time we find a~$\epsilon_k$ that yields the descent criterion~$\Loss(\elt_k-t_k \tilde{g}(\elt_k,\epsilon_k)) < \Loss(\elt_k)-\beta t_k\|\tilde{g}(\elt_k,\epsilon_k)\|^2$, and the parameter~$\gamma$ is no longer needed. 
A careful reading of the proofs provided in the following section yields that the convergence rate (\cref{thm:rate_convergence_algo}) of the resulting algorithm is unchanged, however the asymptotic convergence (\cref{thm:convergence_algo}), case $\bf{(b)}$, needs to be weakened: converging subsequences converge to $\epsilon$-stationary points instead of stationary points. We omit details for the sake of concision. 
\end{remark}
\subsection{Convergence}
\label{sec:convergence}
We show convergence of \cref{alg:minimization} towards stationary points in \cref{thm:convergence_algo}. Finally, \cref{thm:rate_convergence_algo} provides a non-asymptotic sub linear convergence rate, which is by no mean tight yet gives a first estimate of the number of iterations required in order to reach an approximate stationary point. 
\begin{theorem}
\label{thm:convergence_algo}
If~$\eta=0$, then almost surely the algorithm either $\bf{(a)}$ converges in finitely many iterations to an $\epsilon$-stationary point, or $\bf{(b)}$ produces a bounded sequence of iterates~$(\elt_k)_k$ whose converging subsequences all converge to stationary points. 
\end{theorem}
%

%
As an intermediate result, we first show that the update step computed in \cref{alg:update step} is obtained after finitely many iterations and estimate its magnitude relatively to the norm of the descent direction. 
\begin{lemma}
\label{lemma:update_step_magnitude_epsilon}
$\AlgoUpdateStep(\Loss,\elt_k,\epsilon,\eta,\ControlConst_k, \beta, \gamma)$ terminates in finitely many iterations. In addition, let~$\lipconst$ be a Lipschitz constant for the restricted gradients~$\nabla f_i$ (as in \cref{prop:approximation_goldstein_generalized_gradient}) in the~$\epsilon$-neighborhood of~$\elt_k$. Assume that~$\frac{1-\beta}{2\lipconst}\leqslant \ControlConst_{k}$. If~$x_k$ is not an~$(\epsilon,\eta)$-stationary point, then the returned exploration radius~$\epsilon_k$ satisfies:
\[
\min \left( \frac{\gamma^2(1-\beta)}{2\lipconst} \|\tilde{g}(\elt_k,\frac{1}{\gamma}\epsilon_k)\|, \epsilon\right) \leqslant \epsilon_k  \leqslant \min( C_{k} \|\tilde{g}(\elt_k,\epsilon_k)\| , \epsilon).\]
Moreover the returned controlling constant~$\ControlConst_{k+1}$ satisfies:
\[ \ControlConst_{k+1}\geq \frac{\gamma(1-\beta)}{2\lipconst}.\]
\end{lemma}
\begin{proof}
If~$\elt_k$ is a stationary point, then~$\AlgoGeneralizedGradient$ returns a trivial descent direction~$g_k=0$ because the approximate gradient~$G_k$ contains~$\nabla \Loss(\elt_k)$ (Line~1). In turn, $\AlgoUpdateStep$ terminates at Line 4. 

Henceforth we assume that~$\elt_k$ is not a stationary point and that the breaking condition of Line~4 in \cref{alg:update step} is never reached (otherwise the algorithm terminates). Therefore, at each iteration of the main loop, either~$\ControlConst_{k+1}$ is replaced by~$\gamma \ControlConst_{k+1} $ (line 9), or~$\epsilon_k$ is replaced by~$\gamma \epsilon_k$ (line 12), until both the following inequalities hold (line 14):
\[\text{ {\bf (A)} } \, \,  \epsilon_k< \ControlConst_{k+1}  \|\tilde{g}(\elt_k,\epsilon_k)|\; \, \,  \text{  and }   \]
\[\text{    {\bf (B)}  } \, \,  \Loss(\elt_k-t_k\tilde{g}(\elt_k,\epsilon_k))< \Loss(\elt_k)-\beta t_k\|\tilde{g}(\elt_k,\epsilon_k)\|^2. \]
Once~$\ControlConst_{k+1}$ becomes lower than~$\frac{1-\beta}{2\lipconst}$, inequality {\bf (A)} implies inequality {\bf (B)} by \cref{prop:approx_descent} ${\bf (ii)}$. It then takes finitely many replacements~$\epsilon_k \gets \gamma \epsilon_k$ to reach inequality {\bf (A)}, by \cref{prop:approx_descent} {\bf (i)}. At that point (or sooner), \cref{alg:update step} terminates. This concludes the first part of the statement, namely~$\AlgoUpdateStep$ terminates in finitely many iterations. 

Next we assume that~$\elt_k$ is not an~$(\epsilon,\eta)$-stationary point, which ensures that the main loop of $\AlgoUpdateStep$ cannot break at Line~5. We have the invariant~$\ControlConst_{k+1}\geqslant \gamma \frac{1-\beta}{2\lipconst}$: this is true at initialization ($\ControlConst_{k+1}=\ControlConst_{k}$) by assumption, and in later iterations~$\ControlConst_{k+1}$ is only replaced by~$\gamma \ControlConst_{k+1}$ whenever~{\bf (A)} holds but not~{\bf (B)}, which forces~$\ControlConst_{k+1}\geqslant \frac{1-\beta}{2\lipconst}$ by \cref{prop:approx_descent}~{\bf (ii)}.

At the end of the algorithm, $\epsilon_k \leqslant \ControlConst_{k+1} \|\tilde{g}(\elt_k,\epsilon_k)\| $ by inequality {\bf (A)}, and so we deduce the right inequality~$\epsilon_k \leqslant \min(C_{k} \|\tilde{g}(\elt_k,\epsilon_k)\|,\epsilon)$. 

Besides, if both {\bf (A)} and {\bf (B)} hold when entering the main loop (line 11) for the first time, then~$\epsilon_k=\epsilon$. Otherwise, let us consider the penultimate iteration of the main loop for which the update step is~$\frac{1}{\gamma} \epsilon_k$. Then, either condition {\bf (A)} does not hold, namely $\frac{1}{\gamma}\epsilon_k> \ControlConst_{k+1} \|\tilde{g}(\elt_k,\frac{1}{\gamma}\epsilon_k)\| \geqslant \gamma \frac{1-\beta}{2\lipconst} \|\tilde{g}(\elt_k,\frac{1}{\gamma}\epsilon_k)\|$, or condition {\bf (B)} does not hold, which by the assertion~{\bf (ii)} of \cref{prop:approx_descent} implies~$\frac{1}{\gamma}\epsilon_k\geqslant \frac{1-\beta}{2\lipconst} \|\tilde{g}(\elt_k,\frac{1}{\gamma}\epsilon_k)\|$. In any case, we deduce that
\[
 \epsilon_k \geqslant \min\left( \frac{\gamma^2(1-\beta)}{2\lipconst} \left\|\tilde{g}\left(\elt_k,\frac{1}{\gamma}\epsilon_k\right)\right\|, \epsilon\right).\]
\end{proof}

\begin{proof}[Proof of \cref{thm:convergence_algo}]
We assume that alternative~$\bf{(a)}$ does not happen. By \cref{lemma:update_step_magnitude_epsilon}, \cref{alg:update step} terminates in finitely many iteration and by Line~14 we have the guarantee:
\begin{equation}
\label{eq:theorem_convergence_0}
\forall k\geqslant 0, \, \Loss(\elt_k-t_k\tilde{g}(\elt_k,\epsilon_k))< \Loss(\elt_k)-\beta t_k\|\tilde{g}(\elt_k,\epsilon_k)\|^2.
\end{equation}

The subsequent iterate~$\elt_{k+1}$ is initialized at~$\elt_k-t_k\tilde{g}(\elt_k,\epsilon_k)$ by $\AlgoMakeDiff$ (see \cref{alg:make_differentiable}) and replaced by a sample in a progressively shrinking ball~$B(\elt_k-t_k\tilde{g}(\elt_k,\epsilon_k), r)$ until two conditions are reached. The first condition is that~$\Loss$ is differentiable at~$\elt_{k+1}$, and since~$\DiffSet$ has full measure by Rademacher's Theorem, this requirement is almost surely satisfied in finitely many iterations. The second condition is that
\begin{equation}
\label{eq:theorem_convergence_1}
\forall k\geqslant 0, \, \Loss(\elt_{k+1})< \Loss(\elt_k)-\beta t_k\|\tilde{g}(\elt_k,\epsilon_k)\|^2,
\end{equation}
which by \cref{eq:theorem_convergence_0} and continuity of~$\Loss$ is satisfied in finitely many iterations as well. Therefore $\AlgoMakeDiff$ terminates in finitely many iterations almost surely. In particular, the sequence of iterates~$(\elt_k)_k$ is infinite.

By \cref{eq:theorem_convergence_1} the sequence of iterates' values~$(\Loss(\elt_k))_k$ is decreasing and it must converge by compactness of the sublevel sets below~$\Loss$. Using \cref{eq:theorem_convergence_1}, we obtain:
\begin{equation}
\label{eq:theorem_convergence_2}
\epsilon_k \|\tilde{g}(\elt_k,\epsilon_k)\|=t_k\|\tilde{g}(\elt_k,\epsilon_k)\|^2 \leq \frac{1}{\beta}(\Loss(\elt_k)-\Loss(\elt_{k+1})) \longrightarrow 0^+,
\end{equation}
so that in particular, either $\epsilon_k \to 0$ or $\|\tilde{g}(\elt_k,\epsilon_k)\| \to 0$.
Let also~$\lipconst$ be Lipschitz constant for the restricted gradients $\nabla f_i$ (as in \cref{prop:approximation_goldstein_generalized_gradient}) in the $\epsilon$-offset of the sublevel set $\{\elt, \Loss(\elt)\leqslant  \Loss(\elt_0) \}$. Up to taking~$\lipconst$ large enough, there is another Lipschitz constant~$\lipconst'<\lipconst$ ensuring that
\[\frac{1}{\gamma} \frac{1-\beta}{2\lipconst}\leq   \frac{1-\beta}{2\lipconst'}\leqslant \ControlConst_0.\]

By~\cref{lemma:update_step_magnitude_epsilon}, upon termination of \cref{alg:update step},~$\ControlConst_1\geq \gamma \frac{1-\beta}{2\lipconst'}\geq \frac{1-\beta}{2\lipconst}$. If~$\ControlConst_1\leq \frac{1-\beta}{2\lipconst'}$, the~{\bf (ii)} of \cref{prop:approx_descent} prevents Line 9 in \cref{alg:update step} from ever occurring again, i.e.,~$\ControlConst_k=\ControlConst_1$ is constant in later iterations. Otherwise,~$\ControlConst_1$ satisfies~$\ControlConst_1\geq \frac{1-\beta}{2\lipconst'}$ just like~$\ControlConst_0$. A quick induction then yields:
\[\forall k \geq 0, \, \, \ControlConst_k \geq \frac{1-\beta}{2\lipconst}.\]
Therefore, by~\cref{lemma:update_step_magnitude_epsilon}:
\begin{equation}
\label{eq:epsilon_lower_bound}
\forall k \geqslant 0, \,  \min\left( \frac{\gamma^2(1-\beta)}{2\lipconst} \left\|\tilde{g}\left(\elt_k,\frac{1}{\gamma}\epsilon_k\right)\right\|, \epsilon\right) \leqslant \epsilon_k  \leqslant \min( C_{0} \|\tilde{g}(\elt_k,\epsilon_k)\| , \epsilon).
\end{equation}
In particular, using the rightmost inequality and \cref{eq:theorem_convergence_2}, we get~$\epsilon_k\rightarrow 0^+$. In turn, using the leftmost inequality, we get that
\begin{equation}
    \label{eq:min_gradient_goes_0}
   \left\|\tilde{g}\left(\elt_k,\frac{1}{\gamma}\epsilon_k\right)\right\|\rightarrow 0^+.
\end{equation}
The sequence of iterates~$(\elt_k)_k$ is bounded; up to extracting a converging subsequence, we assume that it converges to some $\elt_*$. Let $\epsilon'>0$. We claim that $0\in \partial_{\epsilon'} \Loss(\elt_*)$, that is $\elt_*$ is $\epsilon'$-stationary. As $\elt_k \to \elt_*$ and $\epsilon_k \to 0$, we have that for~$k$ large enough $B(\elt_k,\frac{1}{\gamma}\epsilon_k)\subseteq B(\elt_*,\epsilon')$, which implies that:
\[\partial_{\frac{1}{\gamma}\epsilon_k} \Loss(\elt_k)\subseteq \partial_{\epsilon'}\Loss(\elt_*).\]
Besides, from \cref{prop:approximation_goldstein_generalized_gradient}, we have $\tilde{\partial}_{\frac{1}{\gamma}\epsilon_k}\Loss(\elt_k)\subseteq \partial_{\frac{1}{\gamma}\epsilon_k} \Loss(\elt_k)$, so that $\tilde{g}(\elt_k,\frac{1}{\gamma}\epsilon_k)\in \partial_{\frac{1}{\gamma}\epsilon_k} \Loss(\elt_k)$. Hence $\tilde{g}(\elt_k,\frac{1}{\gamma}\epsilon_k)\in \partial_{\epsilon'}\Loss(\elt_*)$. In the limit, \cref{eq:min_gradient_goes_0} implies $0\in \partial_{\epsilon'}\Loss(\elt_*)$, as desired. 

Following \cite{burke2005robust}, the intersection of the Goldstein subgradients $\partial_{\epsilon'} \Loss(\elt_*)$, over $\epsilon'>0$, yields $\partial \Loss(\elt_*)$. Hence, $0\in \partial \Loss(\elt_*)$ and~$\elt_*$ is a stationary point.
\end{proof}
\begin{theorem}
\label{thm:rate_convergence_algo}
If~$\eta>0$, then \cref{alg:minimization} produces an $(\epsilon,\eta)$-stationary point using at most $O\left(\frac{1}{\eta \min(\eta,\epsilon)}\right)$ iterations. 
\end{theorem}
\begin{proof}
Assume that~\cref{alg:minimization} has run over~$k$ iterations without producing an~$(\epsilon,\eta)$-stationary point. From \cref{alg:update step} (line 14), \cref{alg:make_differentiable} (Line 2) and the choice $t_j=\frac{\epsilon_j}{\|\tilde{g}(\elt_j,\epsilon_j)\|}$ of update step for~$j\leqslant k$, it holds that $\beta \epsilon_j \|\tilde{g}(\elt_j,\epsilon_j)\| \leq \Loss(\elt_j) - \Loss(\elt_{j+1}) $, and in turn
\[\sum_{j=0}^{k-1} \epsilon_j \|\tilde{g}(\elt_j,\epsilon_j)\| \leqslant \frac{f_0-f^*}{\beta},\]
where~$\Loss_0\defeq\Loss(\elt_0)$ and~$\Loss^*$ is a minimal value of $\Loss$. Besides, using~\cref{lemma:update_step_magnitude_epsilon},~$\epsilon_j$ is either bigger than~$\frac{\gamma^2(1-\beta)}{2\lipconst} \left\|\tilde{g}\left(\elt_j,\frac{1}{\gamma}\epsilon_j\right)\right\|$ or than~$\epsilon$, hence
\[\sum_{j=0}^{k-1} \epsilon_j \|\tilde{g}(\elt_j,\epsilon_j)\| \geqslant k \min_{j< k} \|\tilde{g}(\elt_j,\epsilon_j)\|\times \min_{j< k} \epsilon_j  > k\times \eta \times \min\left(\epsilon, \frac{\gamma^2(1-\beta)}{2\lipconst} \eta\right). \]
The two equations cannot simultaneously hold whenever 
\[k\geqslant \frac{f_0-f^*}{\beta}\times \frac{1}{\eta\min\left(\epsilon, \frac{\gamma^2(1-\beta)}{2\lipconst} \eta\right)},\]
which allows us to conclude.
\end{proof}

\subsection{Approximate distance to strata}
\label{section_approx_distance_strata}
The algorithm and its convergence assume that strata~$\Stratum$ that are~$\epsilon$-close to an iterate~$\elt$ can be detected by the oracle~$\SampleOracle(\elt,\epsilon)$. However in practice computing distances~$d(\elt,\Stratum)$ to sub manifolds may be expansive or even impossible. Instead we can hope for approximate distances~$\hat{d}(\elt,\Stratum)$. Typically when we have an assignment 
\[(\elt,\Stratum)\in \R^n\times  \Stratif \longmapsto \proxypoint \in \Stratum \subseteq \R^n, \, \, \,\text{at our disposal, we can define }   \,  \adist(\elt,\Stratum)\defeq d(\elt,\proxypoint),\]
and replace the accurate oracle~$\SampleOracle(\elt,\epsilon)$ with the following approximate oracle:%
\[\ApproxSampleOracle(\elt,\epsilon):=\big\{ \proxypoint\, |  \,  \Stratum\in \Stratif, d(\elt,\proxypoint)\leq \epsilon \big\}=\big\{ \proxypoint\, |  \,  \Stratum\in \Stratif, \hat{d}(\elt,\Stratum)\leq \epsilon \big\}.\]
Therefore for the purpose of this section we make the following assumption: \emph{To every iterate~$\elt\in \R^n$ and stratum~$\Stratum$ we can associate an element~$\proxypoint$ that belongs to~$\Stratum$, in particular we have the corresponding oracle~$\ApproxSampleOracle$. Moreover there exists a constant $\consta \geqslant 1$ such that the resulting approximate distance to strata~$\adist(\elt,\Stratum)\defeq d(\elt,\proxypoint)$ satisfies:
   \[\forall \elt\in \R^n, \forall \Stratum \in \Stratif, \,  \adist(\elt,\Stratum) \leqslant \consta  d(\elt,\Stratum). \] 
}

Note that we always have a reverse inequality~$\adist(\elt,\Stratum)\geqslant d(\elt,\Stratum)$ since~$\proxypoint\in \Stratum$. In the case of the persistence map this will specialize to~$d(\elt,\Stratum) \leqslant \adist(\elt,\Stratum)\leqslant 2d(\elt,\Stratum)$, that is~$\consta=2$, see \cref{prop:x_pi_good_approx_distance_to_strata}. 
 
We then replace the approximate Goldstein subgradient~$\AppGoldstein(\elt)$ with~$\AppAppGoldstein(\elt)$, defined in the exact same way except that it is computed from strata satisfying~$\adist(\elt,\Stratum)\leqslant \epsilon$, that is,~$\AppAppGoldstein(\elt)$ contains~$\nabla \Loss(\proxypoint)$ for each such stratum. The proof of \cref{prop:approximation_goldstein_generalized_gradient} adapts straightforwardly to the following statement: 
\begin{proposition}
\label{prop:approximation_goldstein_generalized_gradient_adapted}
Let $\elt\in \R^n$ and $\epsilon>0$. Assume that~$\Loss$ is stratifiably smooth. Let~$\lipconst$ be a Lipschitz constant of the gradients~$\nabla f_{i}$ restricted to $\overbar{B}(\elt,\consta\epsilon)\cap \Stratum_i$, where~$\Loss_i$ is some local $C^2$ extension of~$\Loss_{|\Stratum_i}$, and~$ \Stratum_i\in \Stratif_{\elt,\epsilon} $ is top dimensional. Then we have:
\[\AppAppGoldstein(\elt) \subseteq \Goldstein(\elt) \,\text{ and } \, \, \Goldstein(\elt)\subseteq \hat{\partial}_{\consta\epsilon}\Loss(\elt) + \overbar{B}(0,(\consta+1)\lipconst \epsilon).\]
\end{proposition}
\begin{proof}
The inclusion~$\AppAppGoldstein(\elt) \subseteq \Goldstein(\elt)$ is clear. Conversely, let $\nabla_{\Stratum}\Loss(\elt')\in \Goldstein(\elt)$, where~$\Stratum$ is a top-dimensional stratum incident to~$\elt'$ and~$|\elt'-\elt|\leqslant \epsilon$. We then have~$\adist(\elt,\Stratum)\leqslant \consta \epsilon$ and hence~$\proxypoint$ is a point in~$\hat{\partial}_{\consta\epsilon}\Loss(\elt)$ which is~$(\consta+1)\epsilon$-close to~$\elt'$. Therefore~$\nabla_{\Stratum}\Loss(\elt')\in \hat{\partial}_{\consta\epsilon}\Loss(\elt) + \overbar{B}(0,(\consta+1)\lipconst \epsilon)$. The rest of the proof is then conducted as in \cref{prop:approximation_goldstein_generalized_gradient}.
\end{proof}
The vector~$\hat{g}(\elt,\epsilon)$ with minimal norm in~$\AppAppGoldstein(\elt)$ can then serve as the new descent direction in place of~$\tilde{g}(\elt,\epsilon)$: 
\begin{proposition}
\label{prop:approx_descent_adapted}
Let~$\Loss$ be stratifiably smooth, and $\elt$ be a non strationnary point. Let $0<\beta<1$, and $\epsilon_0>0$. Denote by~$\lipconst$ a Lipschitz constant for all gradients of the restriction $\Loss_i$ in the ball $\overbar{B}(\elt,\consta\epsilon_0)$ (as in \cref{prop:approximation_goldstein_generalized_gradient_adapted}).
\begin{itemize}
\item[{\bf (i)}] For $0<\epsilon\leqslant \epsilon_0$ small enough we have~$\epsilon\leqslant \frac{1-\beta}{2\lipconst} \|\hat{g}(\elt,\epsilon)\|$; and 
\item[{\bf (ii)}] For such~$\epsilon$, we have~$\forall t \leqslant \frac{\epsilon}{\consta \|\hat{g}(\elt,\epsilon)\|}, \, \Loss(\elt-t\hat{g}(\elt,\epsilon))\leqslant \Loss(\elt)-\beta t\|\hat{g}(\elt,\epsilon)\|^2$.
\end{itemize} 
\end{proposition}
\begin{proof}
The proof of \cref{prop:approx_descent} can be replicated by replacing~$\epsilon$ with~$\frac{\epsilon}{\consta}$ and using \cref{prop:approximation_goldstein_generalized_gradient_adapted} instead of \cref{prop:approximation_goldstein_generalized_gradient}.
\end{proof}
In light of this result, we can use~$g_k=\hat{g}(\elt_k,\epsilon_k)$ as a descent direction, which in practice simply amounts to replace the accurate oracle~$\SampleOracle(\elt_k,\epsilon_k)$ in \cref{alg:generalized gradient} with the approximate oracle~$\ApproxSampleOracle(\elt_k,\epsilon_k)$. 
The only difference is that the assignment of update step in \cref{alg:update step} (Line~7) should take the constant~$\consta$ into account, namely:
\[\text{(Line~7')   } \, \, \, \, \, \, \, \,  t_k\gets \frac{\epsilon}{\consta \|g_k\|}.\]
The convergence analysis of \cref{sec:convergence} holds as well for this algorithm, and the proofs of \cref{thm:convergence_algo} and \cref{thm:rate_convergence_algo} are unchanged.

\section{Application to Topological Data Analysis}
\label{sec:persistence_optim}

\begin{figure}
    \centering
    \includegraphics[width=\textwidth]{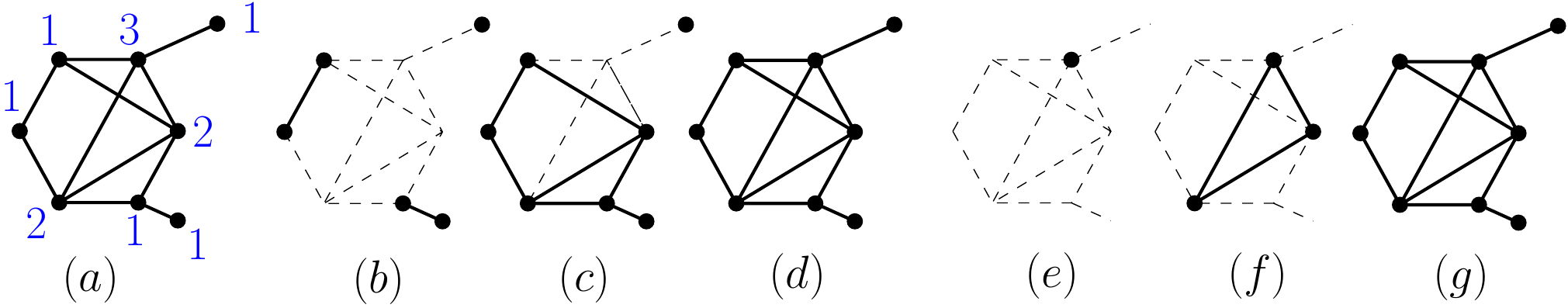}
    \caption{Sublevel sets and superlevel sets filtrations illustrated on graphs. $(a)$ Input graph (V,E) along with the values of a function
$\filt : V \to \R$ (blue). $(b,c,d)$ Sublevel sets for t = 1,2,3 respectively. $(e,f,g)$ Superlevel sets for t = 3,2,1 respectively.}
    \label{fig:illu_filtration}
\end{figure}

In this section, we define the persistence map~$\persmap:\R^n \rightarrow \Barc$ which is a central descriptor in TDA that gives rise to prototypical stratifiably smooth objective functions~$f$ in this work. 
We refer the reader to~\cite{edelsbrunner2008persistent,Oudot2015persistence,zomorodian2005computing} for full treatments of the theory of Persistence. 
We then introduce the stratification that makes~$\persmap$ a stratifiably smooth map, by means of the permutation group. 
Using the associated Cayley graph we give a way to implement the oracle~$\SampleOracle(\elt,\epsilon)$ that samples points in nearby top dimensional strata, which is the key ingredient of \cref{alg:generalized gradient} for computing descent directions.
\subsection{The Persistence Map}
\label{sec:prelim_persistence}
\paragraph{Persistent Homology and Barcodes} Let $n \in \N$, and let~$\{v_1,\dots,v_n\}$ be a (finite) set of vertices. 
A \emph{simplicial complex}~$\SComplex$ is a subset of the power set~$\mathcal{P}(\{v_1,\dots,v_n\})$ whose elements are called \emph{simplices}, and which is closed under inclusion: if~$\simplex\in \SComplex$ is a simplex and~$\simplex'\subseteq \simplex$, then $\simplex'\in \SComplex$. 
The {\em dimension} of the complex is the maximal cardinality of its simplices minus one. In particular a~$1$-dimensional simplicial complex is simply an undirected graph.

A {\em filter function} is a function on the vertices of~$\SComplex$, which we equivalently view as a vector~$\filt\in \R^n$. 
Given~$t\in \R$, we have the sub complexes~$\SComplex_{\leq t} = \{ \simplex \in \SComplex,\ \forall v \in \simplex, \ \filt(v) \leq t\}$.
This yields a nested sequence of sub complexes called the \emph{sublevel sets} filtration of $\SComplex$ by $\filt$: 
\begin{equation}
\label{eq:ordinary_filtration}
\begin{tikzcd}[column sep= normal]
\emptyset \arrow[r] & \cdots \arrow[r] & \SComplex_{\leq s} \arrow[r, "s \leq t"] & \SComplex_{\leq t} \arrow[r] & \cdots \arrow[r]  & \SComplex,
\end{tikzcd}
\end{equation}
See \cref{fig:illu_filtration} for an illustration on graphs. The \emph{(Ordinary) Persistent Homology} of~$\filt$ in degree~$p\in \{0,\cdots,\dim \SComplex\}$ records topological changes in~\cref{eq:ordinary_filtration} by means of points~$(b,d)\in \R^2$, here~$b<d$, called \emph{intervals}. For instance, in degree~$p=0$, an interval~$(b,d)$ corresponds to a connected component that appears in~$\SComplex_{\leq b}$ and that is merged with an older component in~$\SComplex_{\leq d}$. In degree~$p=1$ and~$p=2$, intervals track loops and cavities respectively, and more generally an interval~$(b,d)$ in degree~$p$ tracks a $p$-dimensional sphere that appears in~$\SComplex_{\leq b}$ and persists up to~$\SComplex_{\leq d}$.  

Note that there are possibly infinite intervals~$(b,\infty)$ for $p$-dimensional cycles that persist forever in the filtration~\cref{eq:ordinary_filtration}. Such intervals are not easy to handle in applications, and it is common to consider the \emph{(Extended) Persistent Homology} of~$\filt$, for which they do not occur, i.e. we append the following sequence of pairs involving superlevel sets~$K_{\geq t} \defeq \{\sigma \in K |\ \forall v \in \sigma, \ \filt(v) \geq t \}$ to~\cref{eq:ordinary_filtration}:
\begin{equation}
\label{eq:ordinary_filtration_superlevel}
\begin{tikzcd}[column sep= normal]
\SComplex\cong (\SComplex,\emptyset) \arrow[r] & \cdots \arrow[r] & (\SComplex,\SComplex_{\geq s}) \arrow[r, "s \geq t"] & (\SComplex,\SComplex_{\geq t}) \arrow[r] & \cdots \arrow[r]  & (\SComplex,\SComplex).
\end{tikzcd}
\end{equation}
Together intervals~$(b,d)$ form the (extended)~\emph{barcode}~$\persmap_p(\filt)$ of~$\filt$ in degree~$p$, which we simply denote by~$\persmap(\filt)$ when the degree is clear from the context.   
\begin{definition}
\label{def:barcode}
A {\em barcode} is a finite multi-set of pairs~$(b,d)\in \R^2$ called {\em intervals}, with~$b\leqslant d$. Two barcodes differing by intervals of the form~$(b,b)$ are identified. 
We denote by~$\Barc$ the set of barcodes.
\end{definition}
The set~$\Barc$ of barcodes can be made into a metric space as follows. Given two barcodes~$D \defeq \{(b,d)\}$ and~$D'\defeq\{(b',d')\}$, a \emph{partial matching}~$\gamma: D\rightarrow D'$ is a bijective map from some subset $A \subseteq D$ to some $B \subseteq D'$. 
For~$q\geq 1$ the \emph{$q$-th diagram distance}~$W_q(D,D')$ is the following cost of transferring intervals~$(b,d)$ to intervals~$(b',d')$, minimized over partial matchings~$\gamma$ between $D$ and $D'$:
\begin{equation}\label{eq:def_W_q}
\begin{aligned}
    W_q(D,D')\defeq \inf_{\gamma}   \bigg( &\sum_{(b,d)\in A} \|\gamma(b,d)-(b,d)\|_2^q \\
    &+  \sum_{(b,d)\in D \backslash A} \left(\frac{d-b}{\sqrt{2}}\right)^q + \sum_{(b',d') \in D' \backslash B}  \left(\frac{d'-b'}{\sqrt{2}}\right)^q   \bigg)^{\frac{1}{q}}. 
\end{aligned}
\end{equation}
In particular the intervals that are not in the domain~$A$ and image~$B$ of~$\gamma$ contribute to the total cost relative to their distances to the diagonal~$\{b=d\} \subset \R^2$. 

The Stability Theorem~\cite{cohen2007stability,cohen2010lipschitz} implies that the map~$\persmap:\R^n\rightarrow \Barc$, which we refer to as the \emph{persistence map} in what follows, is Lipschitz continuous.

\paragraph{Differentiability of Persistent Homology} Next we recall from~\cite{leygonie2019framework} the notions of differentiability for maps in and out of~$\Barc$ and the differentiability properties of~$\persmap$. Note that the results of~\cite{leygonie2019framework} focus on ordinary persistence, yet they easily adapt to extended persistence, see e.g.~\cite{yim2021optimisation}. 

Given~$r\in \N$, we define a {\em local $\mathcal{C}^r$-coordinate system} as a collection of~$C^r$ real-valued maps coming in pairs~$b_i,d_i: U\rightarrow \R$ defined on some open Euclidean set~$U\subseteq \R^n$, indexed by a finite set~$I$, and satisfying~$b_i(\elt)\leq d_i(\elt)$ for all~$\elt\in  U$ and~$i\in I$. A local $\mathcal{C}^r$-coordinate system is thus equally represented as a map valued in barcodes
\[\tilde{B}:\elt \in U \mapsto \big\{b_i(x),d_i(x)\big\}_{i\in I}\in \Barc,\] 
where each interval~$(b_i(x),d_i(x))$ is identified and tracked in a~$\mathcal{C}^r$ manner.
%
\begin{definition}
\label{def:differentiability_in_barcode}
A map~$B: \R^n\rightarrow \Barc$ is {\em $r$-differentiable} at~$\elt\in \R^n$ if~$B_{|U}=\tilde{B}_{|U}$ for some local $\mathcal{C}^r$-coordinate system~$\tilde{B}$ defined in a neighborhood~$U$ of~$\elt$.
\end{definition}
Similarly,
\begin{definition}
\label{def:differentiability_out_barcode}
A map~$\V: \Barc \rightarrow \R^m$ is {\em $r$-differentiable} at~$D\in \Barc$ if ~$\V \circ \tilde{B}:\R^n\rightarrow \R^m$ is of class~$\mathcal{C}^r$ in a neighborhood of the origin for all~$n\in \N$ and local $\mathcal{C}^r$-coordinate system~$\tilde{B}$ defined around the origin such that~$\tilde{B}(0)=D$.
\end{definition}
These notions compose together via the chain rule, so for instance an objective function~$f=\V\circ B: \R^n \rightarrow \R^m$ is differentiable in the usual sense as soon as~$B$ and~$\V$ are so. 

We now define the stratification~$\Stratif$ of~$\R^n$ such that the persistence map~$B=\persmap$ is $r$-differentiable (for any~$r$) over each stratum. Denote by~$\Perm$ the group of permutations on~$\{1,\cdots,n\}$. Each permutation~$\pi\in \Perm$ gives rise to a closed polyhedron
\begin{equation}
\label{eq:stratum_permutation}
\StratumPi\defeq\bigg\{ x\in \R^{n} \, | \, \forall 1\leqslant i < n, \, x_{\pi(i)}\leqslant x_{\pi(i+1)} \bigg\},
\end{equation}
which is a \emph{cell} in the sense that its (relative) interior is a top-dimensional stratum of our stratification~$\Stratif$. The (relative) interiors of the various faces of the cells~$\StratumPi$ form the lower dimensional strata. In terms of filter functions, a stratum is simply a maximal subset whose functions induce the same pre-order on vertices of~$\SComplex$. We then have that any persistence based loss is stratifiably smooth w.r.t. this stratification.
\begin{proposition}
\label{prop:persistence_stratifiably_smooth}
Let~$\V: \Barc\rightarrow \R$ be a $2$-differentiable map. Then the objective function~$\Loss\defeq\V\circ \persmap$ is stratifiably smooth for the stratification~$\Stratif$ induced by the permutation group~$\Perm$. 
\end{proposition}
\begin{proof}
From Proposition~4.23 and Corollary~4.24 in~\cite{leygonie2019framework}, on each a cell~$\StratumPi$ we can define a local $C^2$ coordinate system that consists of linear maps~$b_i,d_i: \StratumPi\rightarrow \R $, in particular it admits a~$C^2$ extension on a neighborhood of~$\StratumPi$. Since~$\V$ is globally~$2$-differentiable, by the chain rule, we incidentally obtain a local $C^2$ extension~$f_i$ of~$f_{|\StratumPi}=(\V\circ \persmap)_{|\StratumPi}$.
\end{proof}
\begin{remark}
\label{rke:proper}
Note that the condition that $\Loss=\V\circ \persmap$ is a proper map, as required for the analysis of \cref{alg:minimization}, sometimes fails because~$\persmap$ may not have compact level-sets. The intuitive reason for this is that a filter function~$\filt$ can have an arbitrarily large value on two distinct entries---one simplex creates a homological cycle that the other destroys immediately---that may not be reflected in the barcode~$\persmap(\filt)$. Hence the fiber of~$\persmap$ is not bounded. However, when the simplicial complex~$K$ is (homeomorphic to) a compact oriented manifold, any filter function must reach its maximum at the simplex that generates the fundamental class of the manifold (or one of its components), hence~$\persmap$ has compact level-sets in this case.  
Finally, we note that it is always possible to turn a loss function~$\Loss$ based on~$\persmap$ into a proper map by adding a regularization term that controls the norm of the filter function~$\filt$. 
\end{remark}
\subsection{Exploring the space of strata}
\label{sec:exploration_strata_persistence}
In the setting of \cref{prop:persistence_stratifiably_smooth}, the objective function~$f=\V\circ \persmap:\R^n \rightarrow \R$ is a stratifiably smooth map, where the stratification involved is induced by the group~$\Perm$ of permutations on~$\{1,\cdots,n\}$, with cells~$\StratumPi$ as in \cref{eq:stratum_permutation}. In order to calculate the approximate subgradient~$\Goldstein (\elt)$, we need to compute the set~$\Stratif_{\elt,\epsilon}$ of cells~$\StratumPi$ that are at Euclidean distance no greater than~$\epsilon$ from~$\elt$:
\begin{equation}
\label{eq:stratum_notfar}
d_2^2(\elt,\StratumPi)\defeq\min_{p\in \StratumPi} \|\elt-p\|_2^2\leqslant \epsilon^2.
\end{equation}
\paragraph{Estimating distances to strata} In practice however, solving the quadratic problem of \cref{eq:stratum_notfar} to compute~$d_2(\elt,\StratumPi)$ can be done in $O(n \log n)$ time using solvers for isotonic regression~\cite{best1990active}. 
Since we want to approximate many such distances to neighboring cells, we rather propose the following estimate which boils down to $O(1)$ computations to estimate $d_2(\elt,\StratumPi)$. 
For any~$\pi\in \Perm$, we consider the \emph{mirror of~$\elt$ in~$\StratumPi$}, denoted by~$\elt^{\pi}\in \R^n$ and obtained by permuting the coordinates of~$\elt$ according to~$\pi$:
\begin{equation}
\label{eq:x_pi}
\forall 1\leqslant i \leqslant n, \qquad \elt^{\pi}_{\pi(i)}\defeq\elt_{i}.
\end{equation}
In the rest of this section, we assume that the point~$\elt$ is fixed and has increasing coordinates, $\elt_i\leqslant \elt_{i+1}$, which can always be achieved after a suitable re-ordering of these coordinates. The proxy $d_2(\elt,\elt^\pi)$ then yields a good estimate of~$d_2(\elt,\StratumPi)$, as expressed by the following result.
\begin{proposition}
\label{prop:x_pi_good_approx_distance_to_strata}
For any permutation~$\pi\in \Perm$, we have:
\[d_2(\elt,\StratumPi)\leqslant d_2(\elt,\elt^\pi)\leqslant 2d_2(\elt,\StratumPi). \]
\end{proposition} 
\begin{proof}
The left inequality is clear from the fact that~$\elt^{\pi}$ belongs to the cell~$\StratumPi$. To derive the right inequality, let~$\hat{\elt}^\pi$ be the projection of~$\elt$ onto~$\StratumPi$. It is a well-known fact in the discrete optimal transport literature, or alternatively a consequence of \cref{lemma:inversion_distance} below, that
\[d_2(\elt^{\pi},\hat{\elt}^\pi)= \min_{\tau \in \Perm} d_2(\elt^{\tau},\hat{\elt}^\pi), \]
so that in particular $d_2(\elt^{\pi},\hat{\elt}^\pi)\leqslant d_2(\elt,\hat{\elt}^\pi)$. Consequently,
\[d_2(\elt,\elt^{\pi})\leqslant d_2(\elt,\hat{\elt}^\pi)+ d_2(\elt^{\pi},\hat{\elt}^\pi)\leqslant 2 d_2(\elt,\hat{\elt}^\pi)= 2 d_2(\elt,\StratumPi).\]
\end{proof}
Our approximate oracle for estimating the Goldstein subgradient, see \cref{section_approx_distance_strata}, computes the set of mirrors~$\elt^{\pi}$ that are at most~$\epsilon$-away from the current iterate~$\elt\defeq \elt_k$, that is: 
\[\ApproxSampleOracle(\elt,\epsilon)\defeq\big\{\elt^{\pi}\, | \, d_2(\elt,\elt^{\pi})\leq \epsilon, \pi\in \Perm  \big\}.\]
\begin{remark}
\label{remark_avoiding_recomputation}
Recall that the oracle is called several times in \cref{alg:update step} when updating the current iterate~$\elt_k$ with a decreasing exploration radius~$\epsilon_k$. However, for the oracle above we have
\[\ApproxSampleOracle(\elt,\epsilon')\subseteq \ApproxSampleOracle(\elt,\epsilon) \text{ whenever } \epsilon'\leqslant \epsilon,\] 
so that once we have computed~$\ApproxSampleOracle(\elt_k,\epsilon_k)$ for an initial value $\epsilon_k$ and the current $\elt_k$, we can retrieve $\ApproxSampleOracle(\elt_k,\epsilon')$ for any $\epsilon' < \epsilon_k$ in a straightforward way, avoiding re-sampling neighboring points around~$\elt_k$ and computing the corresponding gradient each time~$\epsilon_k$ decreases, saving an important amount of computational resources.  
\end{remark}
\paragraph{Sampling in nearby strata} In order to implement the oracle~$\ApproxSampleOracle(\elt,\epsilon)$, we consider the Cayley graph with permutations~$\Perm$ as vertices and edges between permutations that differ by elementary transpositions (those that swap consecutive elements). In other words, the Cayley graph is the dual of the stratification of filter functions: a node corresponds uniquely to a cell~$\StratumPi$ and an edge corresponds to a pair of adjacent cells. 

We explore this graph starting at the identity permutation using an arbitrary exploration procedure, for instance the Depth-First Search (DFS) algorithm. During the exploration, assume that the current node, permutation~$\pi$, has not yet been visited (otherwise we discard it). If~$d_2(\elt,\elt^{\pi})\leq \epsilon$, then we record the mirror point~$\elt^{\pi}$. Else,~$d_2(\elt,\elt^{\pi})> \epsilon$, and in this case we do not explore the children of~$\pi$. Note that given a child~$\pi'$ of~$\pi$, we retrieve~$\elt^{\pi'}$ and~$d_2(\elt,\elt^{\pi'})$ from~$\elt^{\pi}$ and~$d_2(\elt,\elt^{\pi})$ in~$O(1)$ time. The following result entails that this procedure indeed computes~$\ApproxSampleOracle(\elt,\epsilon)$. 
\begin{proposition}
\label{prop:cayley_exploration}
%
Let~$\pi'\in \Perm$ be a permutation differing from the identity. Then there must be at least one parent~$\pi$ of~$\pi'$ in the Cayley graph such that~$d_2(\elt,\elt_{\pi})\leq d_2(\elt,\elt_{\pi'})$.
\end{proposition}
\Cref{prop:cayley_exploration} is a straight consequence of the following well-known, elementary lemma.

%
\begin{lemma}
\label{lemma:inversion_distance}
Let $\elt,y\in \R^n$ be two vectors whose coordinates are sorted in the same order, namely $\elt_i\leqslant \elt_j \Leftrightarrow y_i\leqslant y_j$. Given~$\pi\in \Perm$ a permutation, let~$\inv(\pi)$ be the set of inversions, i.e. pairs~$(i,j)$ satisfying $(j-i)(\pi(j)-\pi(i))<0$. Then 
\[\inv(\pi)\subseteq \inv(\pi') \Rightarrow \sum (\elt_i-y_{\pi(i)})^2 \leqslant \sum (\elt_i-y_{\pi'(i)})^2.\]
\end{lemma}
\begin{remark}
\label{rke:dijkstra}
For an arbitrary filter function~$\filt$, the computation of the barcode~$\persmap(\filt)$ has complexity~$O(\#\SComplex^3)$, here~$\#\SComplex$ is the number of vertices and edges in the graph~$\SComplex$ (or the number of simplices if~$\SComplex$ is a simplicial complex). In the SGS algorithm we need to compute~$\persmap(\filt_\pi)$ for each cell~$\StratumPi$ near the current iterate~$\filt_k$, which can quickly become too expansive. Below we describe two heuristics that we implemented in some of our experiments (see \cref{subsec:expe_topomean}) to reduce time complexity.

The first method bounds the number of strata that can be explored with a hyper-parameter~$\CardStrata\in \N$, enabling a precise control of the memory footprint of the algorithm. In this case exploring the Cayley graph of~$\Perm$ using Dijkstra's algorithm is desirable, since it allows to retrieve the~$N$ strata that are the closest to the current iterate~$\filt_k$. Note that in the original Dijkstra's algorithm for computing shortest-path distances to a source node, each node is reinserted in the priority queue each time one of its neighbors is visited. However in our case we dispose of the exact distances~$d(\elt,\elt_\pi)$ to the source each time we encounter a new node, permutation~$\pi$, so we can simplify Dijkstra's algorithm by treating each node of the graph at most once. 
The second approach is \emph{memoization}: inside a cell~$\StratumPi$, all the filter functions induce the same pre-order on the~$n$ vertices of~$\SComplex$, hence the knowledge of the barcode~$\persmap(\filt_{\pi})$ of one of its filter functions allows to compute~$\persmap(\filt_{\pi}')$ for any other~$\filt_{\pi}'\in \StratumPi$ in~$O(\#\SComplex)$ time. We can take advantage of this fact by recording the cells~$\StratumPi$ (and the barcode~$\persmap(\filt_{\pi})$ of one filter function~$\filt_{\pi}$ therein) that are met by the SGS (or GS) algorithm during the optimization, thereby avoiding redundant computations whenever the cell~$\StratumPi$ is met for a second time. 
\end{remark}

\section{Experiments}
\label{sec:experiments}
In this section we apply our approach to the optimisation of objective functions based on the persistence map~$\persmap$, and compare it with other methods.
There are two natural classes of objective functions that we can build on top of the barcode~$\persmap(\filt)$. One consists in turning~$\persmap(\filt)$ into a vector using one of the many existing vectorisation techniques for barcodes \cite{bubenik2015statistical,adams2017persistence,chung2019persistence,carriere2020perslay} and then to apply any standard objective function defined on Euclidean vector space. 
In this work we focus on the second type of objective functions which are based on direct comparisons of barcodes by means of metrics $W_q$ on~$\Barc$ as introduced in \cref{sec:persistence_optim}. 

We consider three experimental settings in increasing level of complexity. \Cref{subsec:total_pers} is dedicated to the optimization of an elementary objective function in TDA that allows for explicit comparisons of SGS with other optimization techniques. \Cref{subsec:expe_registration} and \cref{subsec:expe_topomean} introduce two novel topological optimization tasks: that of \emph{topological registration} for translating filter functions between two distinct simplicial complexes, and that of \emph{topological Fréchet mean} for smoothing the \emph{Mapper graphs} built on top of a data set. 

\paragraph{Implementation} Our implementation is done in \texttt{Python 3} and relies on \texttt{TensorFlow}~\cite{abadi2016tensorflow} for automatic-differentiation, \texttt{Gudhi} \cite{maria2014gudhi} for TDA-related computations (barcodes, distances~$W_q$, Mapper graphs), \texttt{cvxpy} \cite{diamond2016cvxpy} for solving the quadratic minimization problem involved in \cref{alg:generalized gradient}. Our implementation handles both ordinary and extended persistence, complexes of arbitrary dimension, and can easily be tuned to enable general objective functions (assuming those are provided in an automatic differentiation framework). Our code is publicly available at \url{\repolink}. 

\subsection{Proof-of-concept: Minimizing total extended persistence}
\label{subsec:total_pers} 

\begin{figure}
    \centering
    \includegraphics[width=\textwidth]{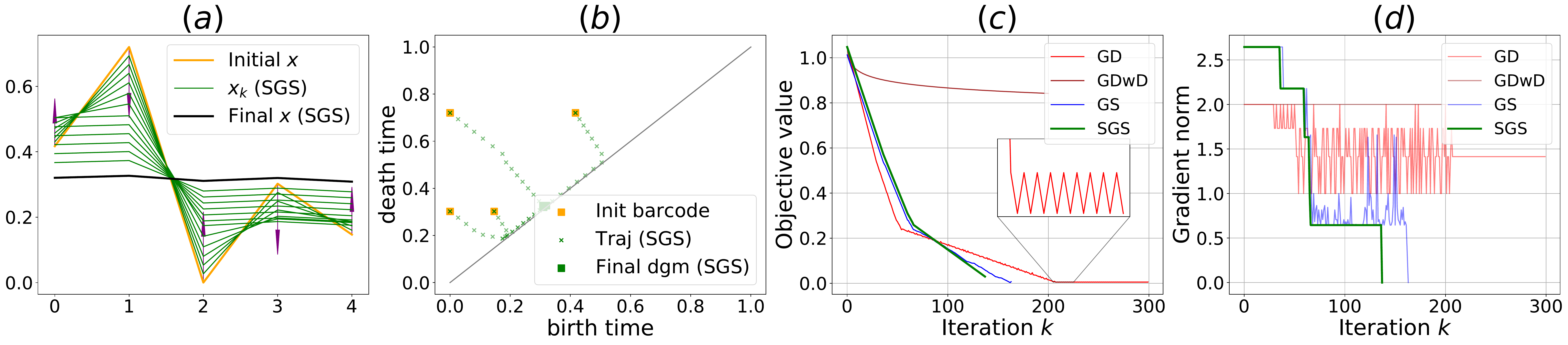}
    \caption{Comparison of vanilla Gradient Descent (GD), Gradient Descent with decay (GDwD), Gradient Sampling (GS) and our Stratified Gradient Sampling (SGS) on a toy example. (a) The evolution of filter functions $(\filt_k)_k$ as the total extended persistence~$\Pers$ is minimised with the SGS method. Purple arrows indicate descent direction at $k=0$. As expected, the minimization tends to make $(\filt_k)_k$ topologically as trivial as possible, that is flat in this context (b) The barcodes~$\persmap(\filt_k)$ represented as \emph{persistence diagrams} extended with the point $(\min(\filt),\max(\filt))$. (c) The value~$\Pers(\filt_k)$ of the objective function across iterations~$k$. (d) The corresponding gradient norms~$(\|g_k\|)_k$. Only GS and SGS reach the stopping criterion $\|g_k\| < \eta$.}
    \label{fig:poc_tda}
\end{figure}

The goal of this experiment is to provide a simple yet instructive framework where one can clearly compare different optimization methods. Here we consider the vanilla Gradient Descent (GD), its variant with learning-rate Decay (GDwD), the Gradient Sampling (GS) methodology and our Stratified Gradient Sampling (SGS) approach. Recall that GD is very well-suited to smooth optimization problems, while GS deals with objective functions that are merely locally Lipschitz with a dense subset of differentiability. To some extent, SGS is specifically tailored to functions with an intermediate degree of regularity since their restrictions to strata are assumed to be smooth, and this type of functions arise naturally in TDA. 

We consider the elementary example of filter functions~$\filt$ on the graph obtained from subdividing the unit interval with~$n$ vertices and the associated (extended) barcodes~$\persmap(\filt)=\persmap_0(\filt)$ in degree~$0$.\footnote{In this setting the extended barcode can be derived from the ordinary barcode by adding the interval $(\min(\filt),\max(\filt))$.} 
When the target diagram is empty, $D=\emptyset$, the objective~$\filt \mapsto W_1(\persmap(\filt),\emptyset)$ to minimize is also known in the TDA literature as the \emph{total extended persistence} of~$\persmap(\filt)$:
\[\Pers: \filt \in \R^n \longmapsto \sum_{(b,d)\in \persmap(\filt)} |d-b| \in \R. \]
Throughout the minimization, the sublevel sets of~$\filt$ are simplified until they become topologically trivial: $\Pers(\filt)$ is minimal if and only if $\filt$ is constant. This elementary optimization problem enables a clear comparison of the GD, GS and SGS methods. 

For each~$\mathrm{mode} \in \{\mathrm{GD},\mathrm{GDwD},\mathrm{GS},\mathrm{SGS}\}$ we get a gradient~$g_k^{\mathrm{mode}}$ and thus build a sequence of iterates 
\[ \filt_{k+1} \defeq \filt_k - \epsilon_k g^\mathrm{mode}_k,\ k\geq 0.\]
For GD, the update step $\epsilon_k = \epsilon$ is constant, for GDwD it is set to be $\epsilon_k = \epsilon/(1+k)$, and for $\mathrm{mode}\in \{\mathrm{GS},\mathrm{SGS}\}$ it is reduced until~$\Pers(x_k - \epsilon_k g_k^\mathrm{mode}) < \Pers(x_k) - \beta \epsilon_k \|g_k^\mathrm{mode}\|^2$ (and in addition $\epsilon_k < C_k \|g_k^\mathrm{SGS}\|$ for SGS). In each case the condition~$\|g_k^\mathrm{mode}\| \leq \eta$ is used as a stopping criterion. 

For the experiments, the graph consists of~$n=5$ vertices,~$x_0 = (0.4, 0.72,0,0.3,0.14)$,~$\epsilon=\eta=0.01$, and we also have the hyper-parameters~$\gamma = 0.5$ and~$\beta = 0.5$ for the GS and SGS algorithm. The results are illustrated in \cref{fig:poc_tda}. 

Whenever differentiable, the objective~$\Pers$ has gradient norm greater than~$1$, so in particular it is \emph{not} differentiable at its minima, which consists of constant functions. Therefore GD oscillates around its optimal value: the stopping criterion~$\|g_k^\mathrm{GD}\| \leq \eta$ is never met which prevents from detecting convergence. Setting $\epsilon_k$ to decay at each iteration in GDwD theoretically ensures the convergence of the sequence $(\filt_k)_k$, but comes at the expense of a dramatic decrease of the convergence rate. 

In contrast, the GS and SGS methods use a fixed step-size $\epsilon_k$ yet they converge since they compute a descent direction by minimizing $\|g\|$ over the convex hull of the surrounding gradients\\ $\{\nabla \Pers(\filt_k), \nabla \Pers(\filt^{(1)}),\dots, \nabla \Pers(\filt^{(m)})\}$, as described in \cref{alg:gradient_sampling} and \cref{alg:generalized gradient}. 
Here~$\filt^{(1)}, \dots, \filt^{(m)}$ are either sampled randomly around the current iterate~$\filt_k$ (with $m=n+1$) for GS or in the strata around~$\filt_k$ (if any) for SGS. We observe that it takes less iterations for SGS to converge:~$137$ iterations versus~$\sim 165$ iterations for GS (averaged over 10 runs). 
This is because in GS the convex hull of the random sampling~$\{\filt^{(1)},\dots,\filt^{(m)} \}$ may be far from the actual generalized gradient~$\Goldstein$, incidentally producing sub-optimal descent directions and missing local minima, while in SGS the sampling takes all nearby strata into account which guarantees a reliable direction (as in \cref{prop:approx_descent}), and in fact since the objective~$\Pers$ restricts to a linear map on each stratum the approximate gradient~$\AppGoldstein(\filt_k)$ equals~$\Goldstein(\filt_k)$.

Another difference is that GS samples~$n+1=6$ nearby points at each iteration~$k$, while SGS samples as many points as there are nearby strata, and for early iterations there is just one such stratum. 
In practice, this results in a total running time of $\sim 2.7$s for GS vs.~$2.4$s for SGS to reach convergence.\footnote{Experiment run on a \texttt{Intel(R) Core(TM) i5-8350U @ 1.70GHz} CPU}    

\subsection{Topological Registration}
\label{subsec:expe_registration}

\begin{figure}
    \centering
    \includegraphics[width=\textwidth]{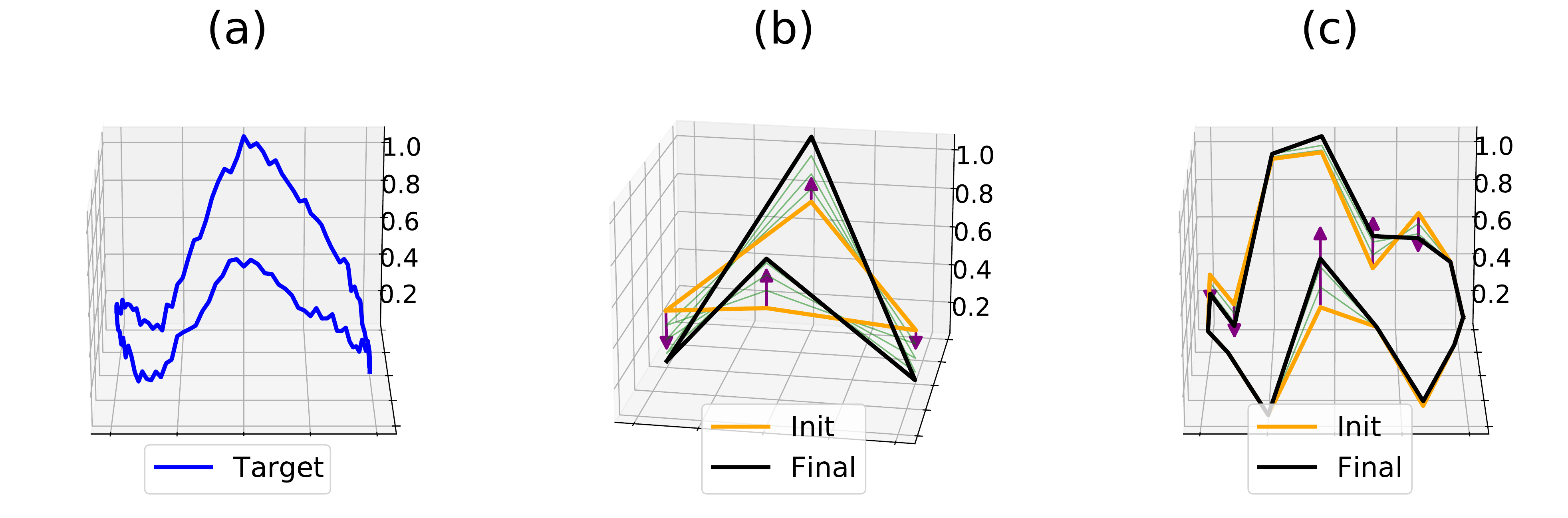}
    \includegraphics[width=0.95\textwidth]{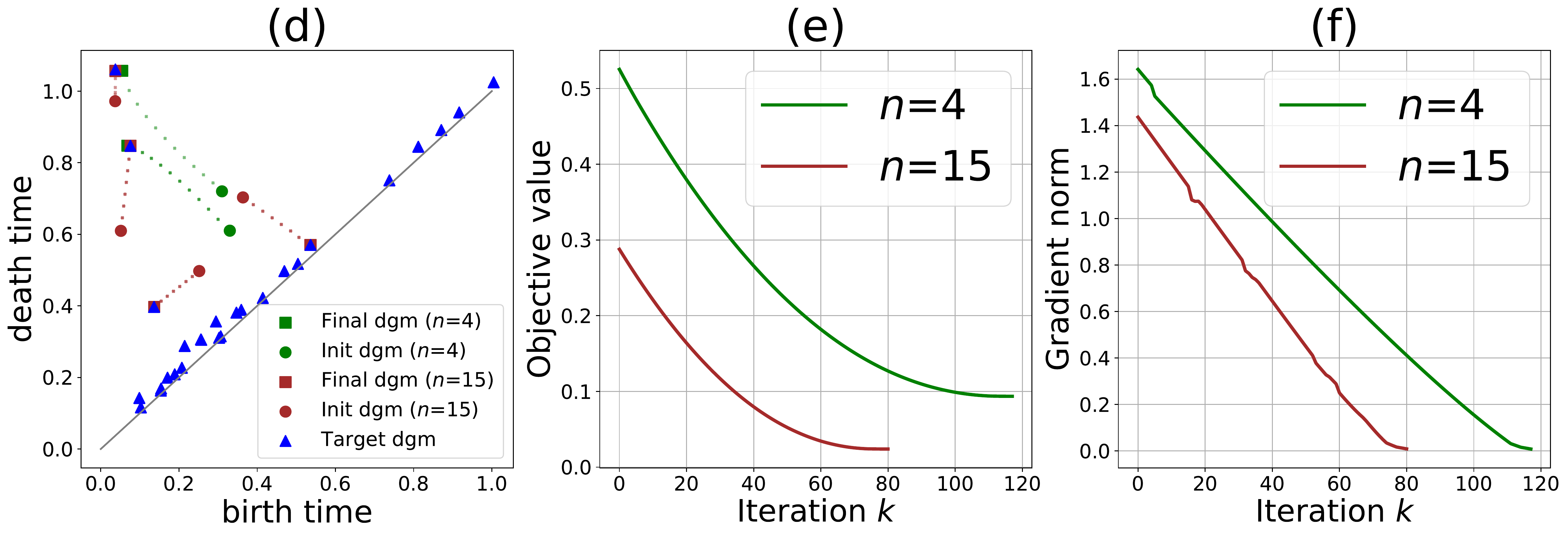}
    \caption{Illustration of topological registration (a) The target function defined on a (circular) simplicial complex of which we want to reproduce the topology. (b) The registration obtained when using a template with $n=4$ vertices. Purple arrows indicate descent direction at $k=0$. (c) The registration obtained when using a template with $n=15$ vertices. (d) The target persistence diagram (blue triangles) along with the diagram trajectories through iterations for both cases (green and brown, respectively). (e) The values of the objective function across iterations. Using a larger template allows to attain lower objective values. (f) The corresponding gradient norms, both reaching the stopping criterion $\|g_k\|\leq \eta$.}
    \label{fig:tR}
\end{figure}

We now present the more sophisticated optimization task of \emph{topological registration}. This problem takes inspiration from registration experiments in shapes and medical images analysis~\cite{joshi2000landmark,ding2001volume,durrleman2011registration}, where we want to translate noisy real-world data (e.g. MRI images of a brain) into a simpler and unified format (e.g. a given template of the brain). 

\paragraph{Problem formulation} In a topological analog of this problem the observation consists of a filter function~$F$ defined on a simplicial complex~$\SComplex$ which may have, for instance, a large number of vertices, and the template consists of a simplicial complex~$\SComplex'$ simpler than $K$ (e.g. with fewer vertices). The goal is then to find a filter function~$\filt$ on~$\SComplex'$ such that~$(K', \filt)$ faithfully recovers the topology of the observation~$(K, F)$.
Formally we minimise the $q$-th distance ($q \in [1,+\infty]$) between their barcodes
\begin{equation}\label{eq:tR}
    x \mapsto W_q(\persmap(x,K'), \persmap(F, K)),
\end{equation}
where we include the complexes in the notations~$\persmap(F, K)$ of the barcodes to make a clear distinction between filter functions defined on~$\SComplex$ and~$\SComplex'$.
\paragraph{Experiment} We minimise~\eqref{eq:tR} using our SGS approach. The observed simplicial complex~$\SComplex$ is taken to be the subdivision of the unit circle with~$120$ vertices, see \cref{fig:tR}. Let $F = F_0 + \zeta$ where $F_0 \in \R^{120}$ is a piecewise linear map satisfying $F_0[0] = 0, F_0[30]=1, F_0[45] = 0.05, F_0[60] = 0.35, F_0[75] = 0.1, F_0[90] = 0.8$ and $\zeta$ is a uniform random noise in $[0,0.1]^{120}$. 
The (extended) barcode of~$F_0$ consists of two long intervals~$(0,1), (0.05,0.9)$ and one smaller interval~$(0.1, 0.35)$ that corresponds to the small variation of~$F_0$ (\cref{fig:tR} (a, left of the plot)). The stability of persistent homology implies that the barcode of~$F$, which is a noisy version of $F_0$, contains a perturbation of these three intervals along with a collection of small intervals of the form $(b,d)$ with $(d-b) < 0.1$, since~$0.1$ is the amplitude of the noise~$\zeta$.
The persistence diagram representation of this barcode can be seen on \cref{fig:tR} (d, blue triangles): the three intervals are represented by the points away from the diagonal $\{x=y\}\subset\R^2$ and the topological noise is accounted by the points close to the diagonal. 

We propose to compute a topological registration~$\filt$ of~$(K,F)$ for two simpler circular complexes with~$n=4$ and $n=15$ vertices respectively (\cref{fig:tR}, (b,c)). 
We initialize the vertex values $\filt_0$ randomly (uniformly in $[0,1]^n$), and minimize \eqref{eq:tR} via SGS. We use $q=2$, and the parameters of \cref{alg:generalized gradient} are set to $\epsilon = 0.01$, $\eta = 0.01$, $\beta=0.5$, $\gamma=0.5$.

With $n=4$ vertices, the final filter function~$\filt$ returned by \cref{alg:generalized gradient} reproduces the two main peaks of~$F$ that correspond to the long intervals~$(0,1), (0.05,0.9)$, but it fails to reproduce the small bump corresponding to~$(0.1,0.35)$ as it lacks the degrees of freedom to do so. \textit{A fortiori} the noise appearing in~$F$ is completely absent in~$\filt$, as observed in \cref{fig:tR} (d) where the two points appearing in the barcode of~$\filt_0$ are pushed towards the two points of the target barcode of~$F$ as it is the best way to reduce the distance~$W_q$. Using $n=15$ vertices the barcode~$\persmap(\filt)$ retrieves the third interval~$(0.1,0.35)$ as well and thus the final filter function~$\filt$ reaches a lower objective value. However~$\filt$ also fits some of the noise, as one of the interval in the diagram of~$x_k$ is pushed toward a noisy interval close to the diagonal (see \cref{fig:tR} (d)). 

\subsection{Topological Mean of Mapper graphs}
\label{subsec:expe_topomean}

In our last experiment, we provide an application of our SGS algorithm to the {\em Mapper} data visualization technique~\cite{Singh2007}.
Intuitively, given a data set~$X$, Mapper produces a graph~$\mapper(X)$, whose attributes, such as its connected components and loops, reflect similar structures in~$X$. For instance, the branches of the Mapper graph in~\cite{Rizvi2017} correspond to the differentiation of stem cells into specialized cells. Besides its potential for applications, Mapper  enjoys strong statistical and topological properties~\cite{Brown2021, Carriere2017, Carriere2018, Carriere2019a, Munch2016, Belchi2020}. 

In this last experiment, we propose an optimization problem to overcome one of the main Mapper limitations, i.e., the fact that Mapper sometimes contains irrelevant features, and solve it with the SGS algorithm. For a proper introduction to Mapper and its main parameters we refer the reader to the appendix (\cref{sec:appendix_mapper}).

\paragraph{Problem formulation} It is a well-known fact that the Mapper graph is not robust to certain changes of parameters which may introduce artificial graph attributes, see~\cite{Attali2009} for an approach to curate~$\mapper(X)$ from its irrelevant attributes. In our case we assume that~$\mapper(X)$ is a graph embedded in some Euclidean space~$\R^d$ ($d=2$ in our experiments), which is typically the case when the data set~$X$ is itself in~$\R^d$, and we modify the embedding of the nodes in order to cancel geometric outliers. For notational clarity we distinguish between the embedded graph~$\mapper(X)\subseteq \R^d$ and its underlying abstract graph~$\SComplex$. Let~$n$ be the number of vertices of~$\SComplex$. 

We propose an elementary scheme inspired from~\cite{Carriere2018} in order to produce a simplified version of~$\mapper(X)$. For this, we consider a family of bootstrapped data sets~$\hat X_1,\dots,\hat X_k$ obtained by sampling the data set~${\rm card}(X)$ times with replacements, from which we derive new mapper graphs~$\SComplex_1,\cdots,\SComplex_k$, whose embeddings~$\mapper(\hat X_1),\dots,\mapper(\hat X_k)$ in~$\R^d$ are fixed during the experiment. In particular, given a fixed unit vector~$\bf e$ in $\R^d$, the projection~$F_{\bf e}$ onto the line parametrized by~${\bf e}$ induces filter functions for each~$\SComplex_i$, hence barcodes~$\persmap(F_{{\bf e}},\SComplex_i)$. 

We minimize the following objective over filter functions~$\tilde{F}_{\bf e}\in \R^n$: 
\begin{equation}\label{eq:frechetmapper_one_direction}
\tilde{F}_{\bf e}\in \R^n \mapsto \sum_{i=1}^k W_2(\persmap (\tilde{F}_{\bf e}, \SComplex),\persmap(F_{\bf e}, \SComplex_i))^2 \in \R. 
\end{equation}
By viewing the optimized filter function~$\tilde{F}_{\bf e}$ as the coordinates of the vertices of~$\SComplex$ along the ${\bf e}$-axis, we obtain a novel embedding of the mapper graph~$\mapper(X)$ in~$\R^d$ that is the \emph{topological barycenter} of the family~$(F_{\bf e}, \mapper(\hat X_i))$. 

To further improve the embedding~$\mapper(X)$, we jointly optimize Eq.~\eqref{eq:frechetmapper_one_direction} over a family $\{{\bf e}_j\}_j$ of directions. Intuitively, irrelevant graph attributes do not appear in most of the subgraphs~$\mapper(\hat X_i)$ and thus are removed in the optimized embedding of~$\mapper(X)$. 
\begin{remark}
\label{remark_frechet_mean}
In some sense, the minimization \cref{eq:frechetmapper_one_direction} corresponds to pulling back to filter functions the well-known minimization problem~$\Barc \ni D \mapsto \sum_{i=1}^k W_2(D,D_i)^2$ that defines the \emph{barycenter} or \emph{Fr\'echet mean} of barcodes~$D_1,\dots,D_k$, see~\cite{turner2014frechet}. Indeed, a \emph{topological mean} of a set of filter functions~$\filt^1,\dots,\filt^k$ on simplicial complexes~$\SComplex_1,\dots, \SComplex_k$ can be defined as a minimizer of~$\filt \in \R^n \mapsto \sum W_2(\persmap(\filt,\SComplex),\persmap(\filt
^i,\SComplex_i))^2$. In our experiment, $\filt$ is interpreted as a radial projection onto the~${\bf e}$-axis, and in fact when considering several directions $\{{\bf e}_j\}_j$ the mean resulting from the optimization is actually that of the so-called Persistent Homology Transform from~\cite{Turner2014a}.
\end{remark}

\paragraph{Experiment} To illustrate this new method for Mapper, we consider a data set~$X$ of single cells characterized by chromatin folding~\cite{Nagano2017}. Each cell is encoded by the squared distance matrix~$M$ of its DNA fragments. This data set was previously studied in~\cite{Carriere2020d}, in which it was shown that the Mapper graph could successfully capture the cell cycle, represented as a big loop in the graph. However, this attribute could only be observed by carefully tuning the parameters. Here we start with a Mapper graph computed out of arbitrary parameters, and then curate the graph using bootstrap iterates as explained in the previous paragraphs.

Specifically, we processed the data set~$X$ with the stratum-adjusted correlation coefficient (SCC)~\cite{Yang2017}, with $500$kb and convolution parameter $h=1$ on all chromosomes. Then, we run a kernel PCA on the SCC matrix to obtain two lenses and computed a Mapper graph from these lenses using resolution~$15$, gain~$0.4$ on both lenses, and hierarchical clustering with threshold~$2$ on Euclidean distance. See \cref{sec:appendix_mapper} for a description of these parameters. The resulting Mapper graph~$\mapper(X)$ displayed in \cref{fig:mapper} (upper left) contains the expected main loop associated to the cell cycle, but it also contains many spurious branches. However computing the Mapper graph with same parameters on a bootstrap iterate results in less branches but also in a coarser version of the graph~(\cref{fig:mapper}, upper middle).

After using the SGS algorithm capped at~$150$ strata (see \cref{rke:dijkstra}), $\epsilon=0.01$, $\eta=0.01$, $\beta=0.5$, $\gamma=0.5$, initialized with $\mapper(X)$, and with loss computed out of $10$ bootstrap iterates and $4$ directions with angles $\{0,\pi/2,\pi/4,-\pi/4\}$, the resulting Mapper, shown in  \cref{fig:mapper} (upper right), offers a good compromise: its resolution remains high and it is curated from irrelevant and artifactual attributes. 

\begin{figure}
    \centering
    \includegraphics[width=0.3\textwidth]{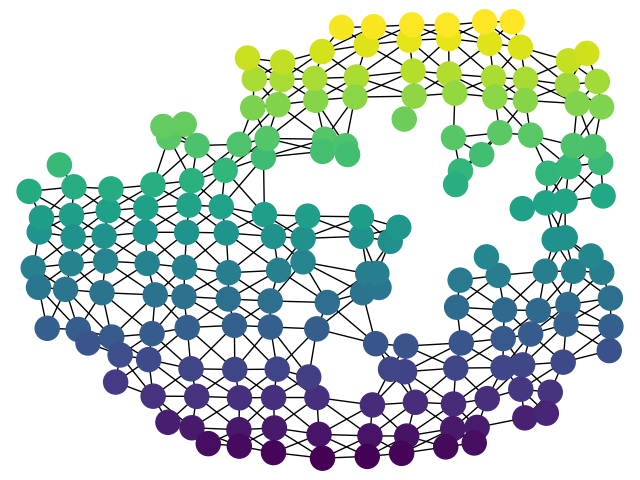}
    \includegraphics[width=0.3\textwidth]{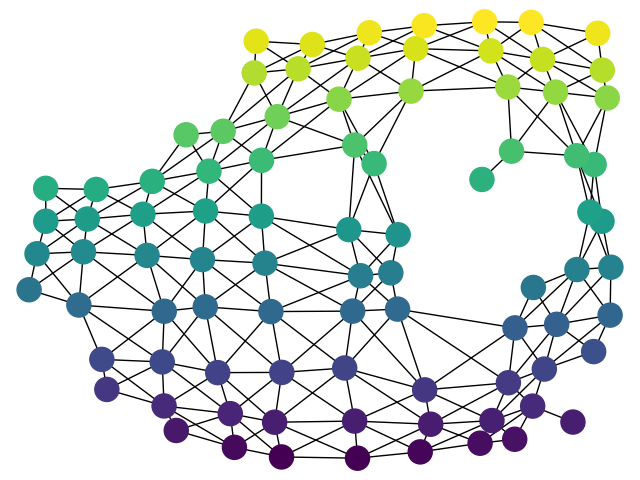}
    \includegraphics[width=0.3\textwidth]{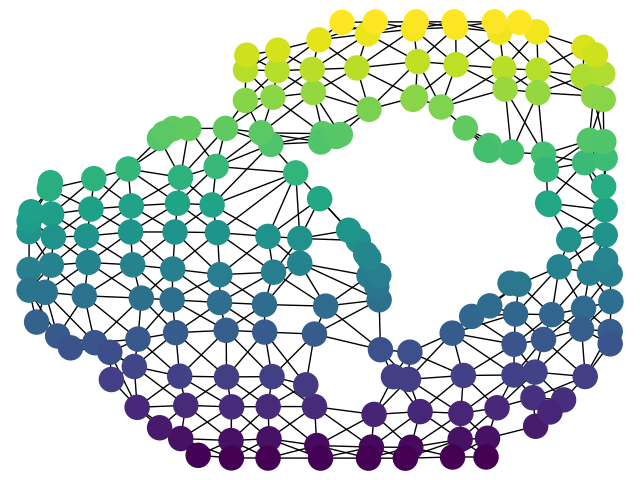}
    \includegraphics[width=0.3\textwidth]{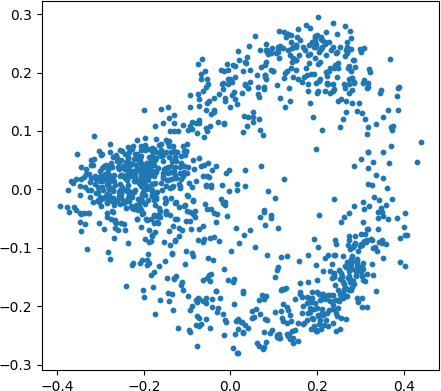}
    \includegraphics[width=0.3\textwidth]{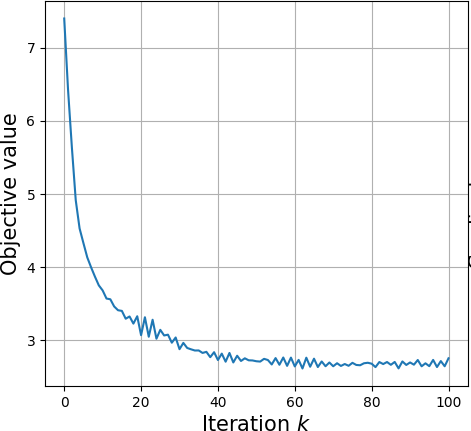}
    \caption{Different Mapper graphs colored with the first kernel PCA component. \emph{(Top row)} Left: original Mapper graph computed with a set of arbitrary parameters with many spurious branches. Middle: Mapper graph obtained from bootstrap with very low resolution. Right: curated Mapper graph obtained as the Fr{\'e}chet mean of the bootstrap iterates. \emph{(Bottom row)} Left: visualization of the data set with kernel PCA. Right: the evolution of the loss \eqref{eq:frechetmapper_one_direction} during the optimization process.}
    \label{fig:mapper}
\end{figure}

\clearpage
\bibliographystyle{plain}
\bibliography{references}

\begin{thebibliography}{10}

\bibitem{abadi2016tensorflow}
Mart{\'\i}n Abadi, Paul Barham, Jianmin Chen, Zhifeng Chen, Andy Davis, Jeffrey
  Dean, Matthieu Devin, Sanjay Ghemawat, Geoffrey Irving, Michael Isard, et~al.
\newblock Tensorflow: A system for large-scale machine learning.
\newblock In {\em 12th $\{$USENIX$\}$ symposium on operating systems design and
  implementation ($\{$OSDI$\}$ 16)}, pages 265--283, 2016.

\bibitem{adams2017persistence}
Henry Adams, Tegan Emerson, Michael Kirby, Rachel Neville, Chris Peterson,
  Patrick Shipman, Sofya Chepushtanova, Eric Hanson, Francis Motta, and Lori
  Ziegelmeier.
\newblock Persistence images: A stable vector representation of persistent
  homology.
\newblock {\em Journal of Machine Learning Research}, 18, 2017.

\bibitem{Attali2009}
Dominique Attali, Marc Glisse, Samuel Hornus, Francis Lazarus, and Dmitriy
  Morozov.
\newblock {Persistence-sensitive simplication of functions on surfaces in
  linear time}.
\newblock In {\em {Topological Methods in Data Analysis and Visualization
  (TopoInVis 2009)}}. Springer, 2009.

\bibitem{attouch2013convergence}
Hedy Attouch, J\'{e}r\^{o}me Bolte, and Benar~Fux Svaiter.
\newblock Convergence of descent methods for semi-algebraic and tame problems:
  proximal algorithms, forward-backward splitting, and regularized
  {G}auss-{S}eidel methods.
\newblock {\em Math. Program.}, 137(1-2, Ser. A):91--129, 2013.

\bibitem{bagirov2014introduction}
Adil Bagirov, Napsu Karmitsa, and Marko~M M{\"a}kel{\"a}.
\newblock {\em Introduction to Nonsmooth Optimization: theory, practice and
  software}.
\newblock Springer, 2014.

\bibitem{Belchi2020}
Francisco Belch{\'{i}}, Jacek Brodzki, Matthew Burfitt, and Mahesan Niranjan.
\newblock {A numerical measure of the instability of Mapper-type algorithms}.
\newblock {\em Journal of Machine Learning Research}, 21(202):1--45, 2020.

\bibitem{bendich2016persistent}
Paul Bendich, James~S Marron, Ezra Miller, Alex Pieloch, and Sean Skwerer.
\newblock Persistent homology analysis of brain artery trees.
\newblock {\em The annals of applied statistics}, 10(1):198, 2016.

\bibitem{bertsekas1975nondifferentiable}
Dimitri~P Bertsekas.
\newblock {\em Nondifferentiable optimization via approximation}.
\newblock Springer, 1975.

\bibitem{best1990active}
Michael~J Best and Nilotpal Chakravarti.
\newblock Active set algorithms for isotonic regression; a unifying framework.
\newblock {\em Mathematical Programming}, 47(1):425--439, 1990.

\bibitem{bihain1984optimization}
A~Bihain.
\newblock Optimization of upper semidifferentiable functions.
\newblock {\em Journal of Optimization Theory and Applications},
  44(4):545--568, 1984.

\bibitem{bolte2007clarke}
J{\'e}r{\^o}me Bolte, Aris Daniilidis, Adrian Lewis, and Masahiro Shiota.
\newblock Clarke subgradients of stratifiable functions.
\newblock {\em SIAM Journal on Optimization}, 18(2):556--572, 2007.

\bibitem{Brown2021}
Adam Brown, Omer Bobrowski, Elizabeth Munch, and Bei Wang.
\newblock {Probabilistic convergence and stability of random Mapper graphs}.
\newblock {\em Journal of Applied and Computational Topology}, 5:99--140, 2021.

\bibitem{bruel2020topology}
Rickard Br{\"u}el-Gabrielsson, Vignesh Ganapathi-Subramanian, Primoz Skraba,
  and Leonidas~J Guibas.
\newblock Topology-aware surface reconstruction for point clouds.
\newblock In {\em Computer Graphics Forum}, volume 39 No. 5, pages 197--207.
  Wiley Online Library, 2020.

\bibitem{bubenik2015statistical}
Peter Bubenik et~al.
\newblock Statistical topological data analysis using persistence landscapes.
\newblock {\em J. Mach. Learn. Res.}, 16(1):77--102, 2015.

\bibitem{burke2019gradient}
James~V Burke, Frank~E Curtis, Adrian~S Lewis, and Michael~L Overton.
\newblock The gradient sampling methodology.
\newblock {\em invited survey for {\em INFORMS Computing Society Newsletter},
  Research Highlights}, 2019.

\bibitem{burke2005robust}
James~V Burke, Adrian~S Lewis, and Michael~L Overton.
\newblock A robust gradient sampling algorithm for nonsmooth, nonconvex
  optimization.
\newblock {\em SIAM Journal on Optimization}, 15(3):751--779, 2005.

\bibitem{carriere2020note}
Mathieu Carriere, Fr{\'e}d{\'e}ric Chazal, Marc Glisse, Yuichi Ike, Hariprasad
  Kannan, and Yuhei Umeda.
\newblock Optimizing persistent homology based functions.
\newblock In {\em International Conference on Machine Learning}, pages
  1294--1303. PMLR, 2021.

\bibitem{carriere2020perslay}
Mathieu Carri{\`e}re, Fr{\'e}d{\'e}ric Chazal, Yuichi Ike, Th{\'e}o Lacombe,
  Martin Royer, and Yuhei Umeda.
\newblock Perslay: A neural network layer for persistence diagrams and new
  graph topological signatures.
\newblock In {\em International Conference on Artificial Intelligence and
  Statistics}, pages 2786--2796. PMLR, 2020.

\bibitem{Carriere2019a}
Mathieu Carri{\`{e}}re and Bertrand Michel.
\newblock {Statistical analysis of Mapper for stochastic and multivariate
  filters}.
\newblock In {\em CoRR}. arXiv:1912.10742, 2019.

\bibitem{Carriere2018}
Mathieu Carri{\`{e}}re, Bertrand Michel, and Steve Oudot.
\newblock {Statistical analysis and parameter selection for Mapper}.
\newblock {\em Journal of Machine Learning Research}, 19(12):1--39, 2018.

\bibitem{Carriere2017}
Mathieu Carri{\`{e}}re and Steve Oudot.
\newblock {Structure and stability of the one-dimensional Mapper}.
\newblock {\em Foundations of Computational Mathematics}, 18(6):1333--1396,
  2017.

\bibitem{Carriere2020d}
Mathieu Carri{\`{e}}re and Ra{\'{u}}l Rabad{\'{a}}n.
\newblock {Topological data analysis of single-cell Hi-C contact maps}.
\newblock In {\em Abel Symposia}, volume~15, pages 147--162. Springer-Verlag,
  2020.

\bibitem{chen2019topological}
Chao Chen, Xiuyan Ni, Qinxun Bai, and Yusu Wang.
\newblock A topological regularizer for classifiers via persistent homology.
\newblock In {\em The 22nd International Conference on Artificial Intelligence
  and Statistics}, pages 2573--2582, 2019.

\bibitem{chung2019persistence}
Yu-Min Chung and Austin Lawson.
\newblock Persistence curves: A canonical framework for summarizing persistence
  diagrams.
\newblock {\em arXiv preprint arXiv:1904.07768}, 2019.

\bibitem{clarke1990optimization}
Frank~H Clarke.
\newblock {\em Optimization and nonsmooth analysis}.
\newblock SIAM, 1990.

\bibitem{clough2019explicit}
James~R Clough, Ilkay Oksuz, Nicholas Byrne, Julia~A Schnabel, and Andrew~P
  King.
\newblock Explicit topological priors for deep-learning based image
  segmentation using persistent homology.
\newblock In {\em International Conference on Information Processing in Medical
  Imaging}, pages 16--28. Springer, 2019.

\bibitem{cohen2007stability}
David Cohen-Steiner, Herbert Edelsbrunner, and John Harer.
\newblock Stability of persistence diagrams.
\newblock {\em Discrete \& Computational Geometry}, 37(1):103--120, 2007.

\bibitem{cohen2010lipschitz}
David Cohen-Steiner, Herbert Edelsbrunner, John Harer, and Yuriy Mileyko.
\newblock Lipschitz functions have l p-stable persistence.
\newblock {\em Foundations of computational mathematics}, 10(2):127--139, 2010.

\bibitem{corcoran2020regularization}
Padraig Corcoran and Bailin Deng.
\newblock Regularization of persistent homology gradient computation.
\newblock {\em arXiv preprint arXiv:2011.05804}, 2020.

\bibitem{dabaghian2014reconceiving}
Yuri Dabaghian, Vicky~L Brandt, and Loren~M Frank.
\newblock Reconceiving the hippocampal map as a topological template.
\newblock {\em Elife}, 3:e03476, 2014.

\bibitem{davis2020stochastic}
Damek Davis, Dmitriy Drusvyatskiy, Sham Kakade, and Jason~D Lee.
\newblock Stochastic subgradient method converges on tame functions.
\newblock {\em Foundations of computational mathematics}, 20(1):119--154, 2020.

\bibitem{diamond2016cvxpy}
Steven Diamond and Stephen Boyd.
\newblock Cvxpy: A python-embedded modeling language for convex optimization.
\newblock {\em The Journal of Machine Learning Research}, 17(1):2909--2913,
  2016.

\bibitem{ding2001volume}
Lijin Ding, Ardeshir Goshtasby, and Martin Satter.
\newblock Volume image registration by template matching.
\newblock {\em Image and Vision Computing}, 19(12):821--832, 2001.

\bibitem{drusvyatskiy2015clarke}
Dmitriy Drusvyatskiy, Alexander~D Ioffe, and Adrian~S Lewis.
\newblock Clarke subgradients for directionally lipschitzian stratifiable
  functions.
\newblock {\em Mathematics of Operations Research}, 40(2):328--349, 2015.

\bibitem{durrleman2011registration}
Stanley Durrleman, Pierre Fillard, Xavier Pennec, Alain Trouv{\'e}, and
  Nicholas Ayache.
\newblock Registration, atlas estimation and variability analysis of white
  matter fiber bundles modeled as currents.
\newblock {\em NeuroImage}, 55(3):1073--1090, 2011.

\bibitem{edelsbrunner2008persistent}
Herbert Edelsbrunner and John Harer.
\newblock Persistent homology-a survey.
\newblock {\em Contemporary mathematics}, 453:257--282, 2008.

\bibitem{fuduli2004dc}
Antonio Fuduli, Manlio Gaudioso, and Giovanni Giallombardo.
\newblock A {DC} piecewise affine model and a bundling technique in nonconvex
  nonsmooth minimization.
\newblock {\em Optimization Methods and Software}, 19(1):89--102, 2004.

\bibitem{fuduli2004minimizing}
Antonio Fuduli, Manlio Gaudioso, and Giovanni Giallombardo.
\newblock Minimizing nonconvex nonsmooth functions via cutting planes and
  proximity control.
\newblock {\em Siam journal on optimization}, 14(3):743--756, 2004.

\bibitem{gabrielsson2020topology}
Rickard~Br{\"u}el Gabrielsson, Bradley~J Nelson, Anjan Dwaraknath, and Primoz
  Skraba.
\newblock A topology layer for machine learning.
\newblock In {\em International Conference on Artificial Intelligence and
  Statistics}, pages 1553--1563. PMLR, 2020.

\bibitem{gameiro2016continuation}
Marcio Gameiro, Yasuaki Hiraoka, and Ippei Obayashi.
\newblock Continuation of point clouds via persistence diagrams.
\newblock {\em Physica D: Nonlinear Phenomena}, 334:118--132, 2016.

\bibitem{goldstein1977optimization}
AA~Goldstein.
\newblock Optimization of lipschitz continuous functions.
\newblock {\em Mathematical Programming}, 13(1):14--22, 1977.

\bibitem{haarala2007globally}
Napsu Haarala, Kaisa Miettinen, and Marko~M M{\"a}kel{\"a}.
\newblock Globally convergent limited memory bundle method for large-scale
  nonsmooth optimization.
\newblock {\em Mathematical Programming}, 109(1):181--205, 2007.

\bibitem{helou2017local}
Elias~Salom{\~a}o Helou, Sandra~A Santos, and Lucas~EA Sim{\~o}es.
\newblock On the local convergence analysis of the gradient sampling method for
  finite max-functions.
\newblock {\em Journal of Optimization Theory and Applications},
  175(1):137--157, 2017.

\bibitem{hiraoka2016hierarchical}
Yasuaki Hiraoka, Takenobu Nakamura, Akihiko Hirata, Emerson~G Escolar, Kaname
  Matsue, and Yasumasa Nishiura.
\newblock Hierarchical structures of amorphous solids characterized by
  persistent homology.
\newblock {\em Proceedings of the National Academy of Sciences},
  113(26):7035--7040, 2016.

\bibitem{hofer2020graph}
Christoph Hofer, Florian Graf, Bastian Rieck, Marc Niethammer, and Roland
  Kwitt.
\newblock Graph filtration learning.
\newblock In {\em International Conference on Machine Learning}, pages
  4314--4323. PMLR, 2020.

\bibitem{hofer2019connectivity}
Christoph Hofer, Roland Kwitt, Marc Niethammer, and Mandar Dixit.
\newblock Connectivity-optimized representation learning via persistent
  homology.
\newblock In {\em International Conference on Machine Learning}, pages
  2751--2760. PMLR, 2019.

\bibitem{hu2019topology}
Xiaoling Hu, Fuxin Li, Dimitris Samaras, and Chao Chen.
\newblock Topology-preserving deep image segmentation.
\newblock In {\em Advances in Neural Information Processing Systems}, pages
  5658--5669, 2019.

\bibitem{ioffe2009invitation}
AD~Ioffe.
\newblock An invitation to tame optimization.
\newblock {\em SIAM Journal on Optimization}, 19(4):1894--1917, 2009.

\bibitem{joshi2000landmark}
Sarang~C Joshi and Michael~I Miller.
\newblock Landmark matching via large deformation diffeomorphisms.
\newblock {\em IEEE transactions on image processing}, 9(8):1357--1370, 2000.

\bibitem{kachan2020persistent}
Oleg Kachan.
\newblock Persistent homology-based projection pursuit.
\newblock In {\em Proceedings of the IEEE/CVF Conference on Computer Vision and
  Pattern Recognition Workshops}, pages 856--857, 2020.

\bibitem{kiwiel2006methods}
Krzysztof~C Kiwiel.
\newblock {\em Methods of descent for nondifferentiable optimization}, volume
  1133.
\newblock Springer, 2006.

\bibitem{kiwiel2007convergence}
Krzysztof~C Kiwiel.
\newblock Convergence of the gradient sampling algorithm for nonsmooth
  nonconvex optimization.
\newblock {\em SIAM Journal on Optimization}, 18(2):379--388, 2007.

\bibitem{leygonie2019framework}
Jacob Leygonie, Steve Oudot, and Ulrike Tillmann.
\newblock A framework for differential calculus on persistence barcodes.
\newblock {\em Foundations of Computational Mathematics}, pages 1--63, 2021.

\bibitem{li2014persistence}
Chunyuan Li, Maks Ovsjanikov, and Frederic Chazal.
\newblock Persistence-based structural recognition.
\newblock In {\em Proceedings of the IEEE Conference on Computer Vision and
  Pattern Recognition}, pages 1995--2002, 2014.

\bibitem{lukvsan1998bundle}
Ladislav Luk{\v{s}}an and Jan Vl{\v{c}}ek.
\newblock A bundle-{N}ewton method for nonsmooth unconstrained minimization.
\newblock {\em Mathematical Programming}, 83(1-3):373--391, 1998.

\bibitem{maria2014gudhi}
Cl{\'e}ment Maria, Jean-Daniel Boissonnat, Marc Glisse, and Mariette Yvinec.
\newblock The gudhi library: Simplicial complexes and persistent homology.
\newblock In {\em International congress on mathematical software}, pages
  167--174. Springer, 2014.

\bibitem{mifflin1977algorithm}
Robert Mifflin.
\newblock An algorithm for constrained optimization with semismooth functions.
\newblock {\em Mathematics of Operations Research}, 2(2):191--207, 1977.

\bibitem{moor2019topological}
Michael Moor, Max Horn, Bastian Rieck, and Karsten Borgwardt.
\newblock Topological autoencoders.
\newblock In {\em International conference on machine learning}, pages
  7045--7054. PMLR, 2020.

\bibitem{Munch2016}
Elizabeth Munch and Bei Wang.
\newblock {Convergence between categorical representations of Reeb space and
  Mapper}.
\newblock In {\em 32nd International Symposium on Computational Geometry (SoCG
  2016)}, volume~51, pages 53:1--53:16. Schloss Dagstuhl--Leibniz-Zentrum fuer
  Informatik, 2016.

\bibitem{Nagano2017}
Takashi Nagano, Yaniv Lubling, Csilla V{\'{a}}rnai, Carmel Dudley, Wing Leung,
  Yael Baran, Netta Mendelson-Cohen, Steven Wingett, Peter Fraser, and Amos
  Tanay.
\newblock {Cell-cycle dynamics of chromosomal organization at single-cell
  resolution}.
\newblock {\em Nature}, 547:61--67, 2017.

\bibitem{noll2014convergence}
Dominikus Noll.
\newblock Convergence of non-smooth descent methods using the
  {K}urdyka-{L}ojasiewicz inequality.
\newblock {\em J. Optim. Theory Appl.}, 160(2):553--572, 2014.

\bibitem{Oudot2015persistence}
Steve~Y Oudot.
\newblock {\em Persistence theory: from quiver representations to data
  analysis}, volume 209.
\newblock American Mathematical Society Providence, RI, 2015.

\bibitem{perea2015sliding}
Jose~A Perea and John Harer.
\newblock Sliding windows and persistence: An application of topological
  methods to signal analysis.
\newblock {\em Foundations of Computational Mathematics}, 15(3):799--838, 2015.

\bibitem{poulenard2018topological}
Adrien Poulenard, Primoz Skraba, and Maks Ovsjanikov.
\newblock Topological function optimization for continuous shape matching.
\newblock In {\em Computer Graphics Forum}, volume 37 No. 5, pages 13--25.
  Wiley Online Library, 2018.

\bibitem{Rizvi2017}
Abbas Rizvi, Pablo C{\'{a}}mara, Elena Kandror, Thomas Roberts, Ira Schieren,
  Tom Maniatis, and Ra{\'{u}}l Rabad{\'{a}}n.
\newblock {Single-cell topological RNA-seq analysis reveals insights into
  cellular differentiation and development}.
\newblock {\em Nature Biotechnology}, 35:551--560, 2017.

\bibitem{shor2012minimization}
Naum~Zuselevich Shor.
\newblock {\em Minimization methods for non-differentiable functions},
  volume~3.
\newblock Springer Science \& Business Media, 2012.

\bibitem{Singh2007}
Gurjeet Singh, Facundo M{\'{e}}moli, and Gunnar Carlsson.
\newblock {Topological methods for the analysis of high dimensional data sets
  and 3D object recognition}.
\newblock In {\em 4th Eurographics Symposium on Point-Based Graphics (SPBG
  2007)}, pages 91--100. The Eurographics Association, 2007.

\bibitem{solomon2020fast}
Yitzchak Solomon, Alexander Wagner, and Paul Bendich.
\newblock A fast and robust method for global topological functional
  optimization.
\newblock In {\em International Conference on Artificial Intelligence and
  Statistics}, pages 109--117. PMLR, 2021.

\bibitem{townsend2020representation}
Jacob Townsend, Cassie~Putman Micucci, John~H Hymel, Vasileios Maroulas, and
  Konstantinos~D Vogiatzis.
\newblock Representation of molecular structures with persistent homology for
  machine learning applications in chemistry.
\newblock {\em Nature communications}, 11(1):1--9, 2020.

\bibitem{turner2014frechet}
Katharine Turner, Yuriy Mileyko, Sayan Mukherjee, and John Harer.
\newblock Fr{\'e}chet means for distributions of persistence diagrams.
\newblock {\em Discrete \& Computational Geometry}, 52(1):44--70, 2014.

\bibitem{Turner2014a}
Katharine Turner, Sayan Mukherjee, and Doug~M. Boyer.
\newblock Persistent homology transform for modeling shapes and surfaces.
\newblock {\em Information and Inference: A Journal of the IMA}, 3(4):310--344,
  2014.

\bibitem{umeda2017time}
Yuhei Umeda.
\newblock Time series classification via topological data analysis.
\newblock {\em Information and Media Technologies}, 12:228--239, 2017.

\bibitem{vlvcek2001globally}
Jan Vl{\v{c}}ek and Ladislav Luk{\v{s}}an.
\newblock Globally convergent variable metric method for nonconvex
  nondifferentiable unconstrained minimization.
\newblock {\em Journal of Optimization Theory and Applications},
  111(2):407--430, 2001.

\bibitem{Yang2017}
Tao Yang, Feipeng Zhang, Galip Yardımcı, Fan Song, Ross Hardison, William
  Noble, Feng Yue, and Qunhua Li.
\newblock {HiCRep: assessing the reproducibility of Hi-C data using a
  stratum-adjusted correlation coefficient}.
\newblock {\em Genome Research}, 27(11):1939--1949, 2017.

\bibitem{yim2021optimisation}
Ka~Man Yim and Jacob Leygonie.
\newblock Optimization of spectral wavelets for persistence-based graph
  classification.
\newblock {\em Frontiers in Applied Mathematics and Statistics}, 7:16, 2021.

\bibitem{zomorodian2005computing}
Afra Zomorodian and Gunnar Carlsson.
\newblock Computing persistent homology.
\newblock {\em Discrete \& Computational Geometry}, 33(2):249--274, 2005.

\end{thebibliography}

\appendix
\section{Background on Mapper}\label{sec:appendix_mapper}

The Mapper is a visualization tool that allows to represent any data set $X$ equipped with a metric and a continuous function $f:X\rightarrow\R$ with a graph. It is based on the Nerve Theorem, which essentially states that, under certain conditions, the nerve of a cover of a space has the same topology of the original space, where a cover is a family of subspaces whose union is the space itself, and the (1-skeleton of a) nerve is a graph whose nodes are the cover elements and whose edges are determined by the intersections of cover elements. The whole idea of Mapper is that since covering a space is not always simple, an easier way is to cover the image of a continuous function defined on the space with regular intervals, and then pull back this cover to obtain a cover of the original space. 

More formally,
Mapper has three parameters: the {\em resolution} $r\in \mathbb N^*$, the {\em gain} $g\in [0,1]$, and a {\em clustering method} $\mathcal C$. Essentially, the Mapper is defined as $\mapper(X)=\mathcal N(\mathcal C(f^{-1}(\mathcal I(r,g))))$, where $\mathcal I(r,g)$ stands for a cover of ${\rm im}(f)$ with $r$ intervals with $g\%$ overlap, and $\mathcal N$ stands for the nerve operation, which is applied on the cover $\mathcal C(f^{-1}(\mathcal I(r,g)))$ of $X$. This cover is made of the connected components (assessed by $\mathcal C$) of the subspaces $f^{-1}(I), I\in \mathcal I(r,g)$. See Figure~\ref{fig:exmapper}.

\begin{figure}[h]
    \centering
    \includegraphics[width=13cm]{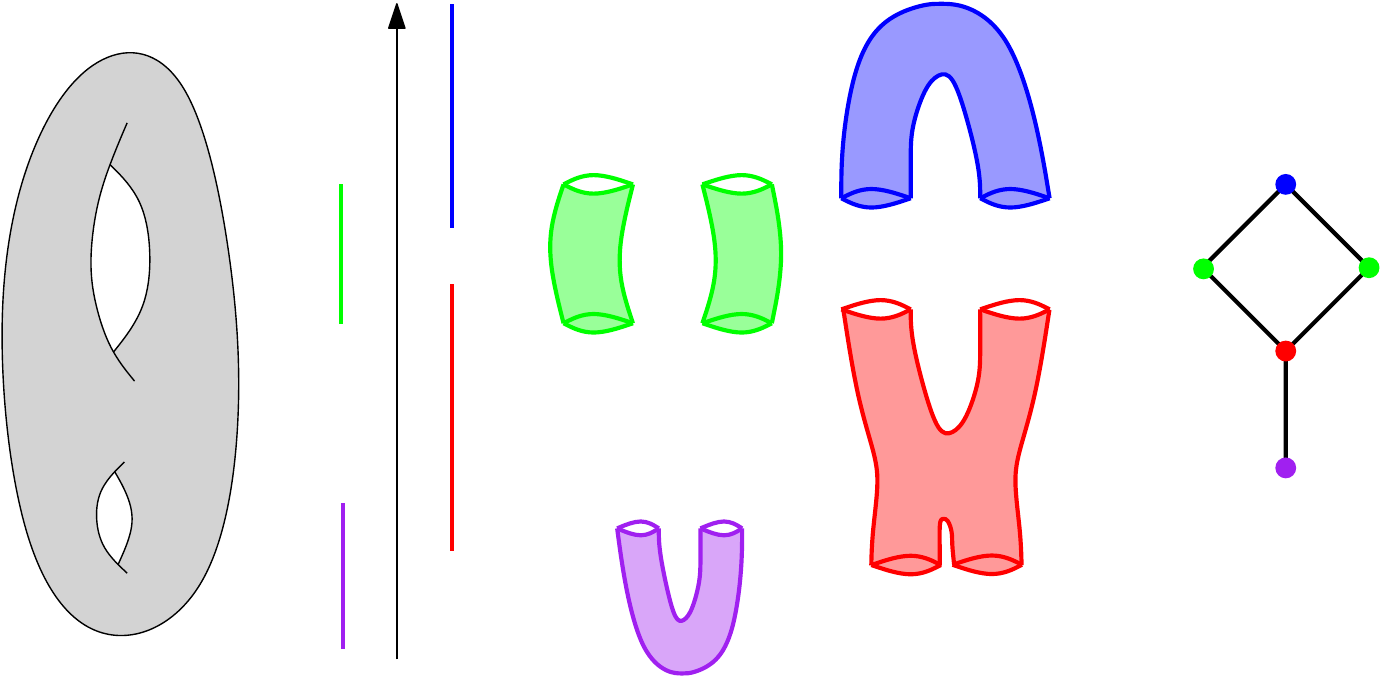}
    \caption{Example of Mapper computation on a double torus with height function covered by four intervals. Ech interval is pulled back in the original space through $f^{-1}$ and then separated into its connected components with $\mathcal C$. In particular, the cover element obtained with the preimage of the green interval is separated into its two connected components. The nerve of this new cover is then computed to obtain the Mapper. One can see that with only four intervals, the topology of the double torus is only partially captured since only one loop is present in the Mapper instead of two.}
    \label{fig:exmapper}
\end{figure}

The influence of the parameters $r,g,\mathcal C,f$ on the Mapper shape is still an active research area. For instance, the number of Mapper nodes increases with the resolution, and the number of edges increases with the gain, but these parameters, as well as the function $f$ and the clustering method $\mathcal C$, can also have more subtle effects on the Mapper shape. We refer  the interested reader to the references mentioned in this article for a more detailed introduction to Mapper.

\end{document}